\documentclass[11pt]{article} 
\usepackage{epsfig} 
\usepackage{amssymb,amsmath,amsthm,amscd}
\usepackage{latexsym} 
\usepackage{mathrsfs}
\usepackage{cancel}
\usepackage[dvipsnames]{xcolor}
\usepackage{amssymb}

\usepackage[margin=1.1in]{geometry}
\usepackage{float}
\usepackage{chngcntr}
\usepackage{graphicx}
\usepackage{multicol}
\usepackage{textcomp}
\counterwithout*{equation}{section} 
\DeclareSymbolFont{matha}{OML}{txmi}{m}{it}
\DeclareMathSymbol{\varv}{\mathord}{matha}{118}
\usepackage{graphicx}
\usepackage[nottoc,numbib]{tocbibind}
 
\newtheorem{thm}{Theorem}[section] 
\newtheorem{prop}[thm]{Proposition} 
 
\newtheorem{lem}[thm]{Lemma}

\theoremstyle{definition}

\theoremstyle{remark} 
\newtheorem{rem}{Remark}[section]  
\def\eqref#1{(\ref{#1})} 

\newcommand {\mat}      [1] {\left[\begin{array}{#1}}
\newcommand {\rix}          {\end{array}\right]}
\newcommand {\de}      [1] {\left|\begin{array}{#1}}
\newcommand {\nt}          {\end{array}\right|}

\newcommand{\bstar}       {\begin{eqnarray*}}
\newcommand{\estar}       {\end{eqnarray*}}
\newcommand{\eqn}       {\begin{eqnarray}}
\newcommand{\enn}       {\end{eqnarray}}
\newcommand{\eq}[1]   {\begin{equation}\label{#1}}
\newcommand{\en}      {\end{equation}}

\begin{document}

\begin{titlepage}  
\title{ Local Hard-Sphere Poisson-Nernst-Planck Models for Ionic Channels with Permanent Charges }  
 \date{\today}
 \author{Weishi Liu\footnote{Department of Mathematics, University of Kansas, 
Lawrence, KS 66045 ({\tt wsliu@ku.edu}).  
} \; and Hamid Mofidi\footnote{Department of Mathematics, University of Iowa, Iowa City, IA 52242 ({\tt hamid-mofidi@uiowa.edu}).} }
\end{titlepage}

\maketitle  


\begin{abstract}
The main goal of this work is to examine the qualitative effect of ion sizes via a steady-state boundary value problem.
We study a one-dimensional version of a Poisson-Nernst-Planck system with a local hard-sphere potential model for ionic flow through a membrane channel with fixed boundary ion concentrations and electric potentials. A complete set of integrals for the inner system is illustrated that delivers information for boundary and internal layers. In addition, a group of simultaneous equations appears in the construction of singular orbits. 
The research aims to set up a simple formation defined by the profile of permanent charges with two mobile ion species, one positively charged, cation, and one negatively charged, anion.
A local hard-sphere potential that depends pointwise on ion concentrations is included in the model to estimate ion-size impacts on the ionic flow.  The analysis is built on the geometric singular perturbation theory, particularly on specific structures of this concrete model.
For 1:1 ionic mixtures, we first conduct precise mathematical analysis and derive a matching system of nonlinear algebraic equations.
We then extend the results by directing on a critical case where the current is zero.
Treating the ion sizes as small parameters, we derive an approximation of zero-current fluxes for general values of permanent charge. We then focus on small values of permanent charge to obtain more concrete outputs. 
We will also examine ion size effects on the flow rate of matter for the zero-current case.
\end{abstract}


\section{Introduction.}
\setcounter{equation}{0}
The migration of ions through open ion channels is one of the most remarkable physical problems performed by living cells. Cells are enveloped by lipid membranes that are almost impermeable to physiological ions (mostly Na$^+$, K$^+$, Ca$^{2+}$, and Cl$^-$). One mechanism for ions to travel across these membranes is through open ion channels, which are proteins found in cell membranes that produce holes in the membrane to enable cells to interact with each other and with the outer surface to transform signals and to conduct tasks together \cite{BNVHEG,Eis00,Hil01,Hille89}.

A primary role of ion channels is to manage the permeability of membranes for a given species of ions, and to choose the types of ions, and to help and modulate the diffusion of ions across cell membranes \cite{CE,Eis,GNE,IR}. One typically determines the permeation and selectivity features of ion channels from the current-voltage (I-V) relations measured experimentally. Individual fluxes transfer more data than the current; however, measuring them has been discouraging and costly \cite{JEL19,SNE01,VGK10}.

Device equations are most useful when they predict complex behaviors realistically while using only a few parameters with fixed values. Fortunately, electrodiffusion allows productive behavior with simple device equations and a fixed set of parameters. One may describe the diverse (technologically important) behavior of transistors by simple conservation laws and constitutive relations, the Poisson-Nernst-Planck (PNP) equations using fixed values of parameters.
The simplest PNP system is the classical PNP (cPNP) system \cite{EL07}. The classical PNP treats dilute ionic mixtures, where no ionto-
ion interactions are involved. More sophisticated models have also been studied in \cite{BKSA09,EHL10,EL17,HBRR18}, etc.

The underlying assumption in the derivation of cPNP (classical PNP) systems is a ``dilute'' hypothesis so that one can treat ions as point charges. In particular, cPNP systems treat ions with the equal valences essentially the same except it associates different diffusion coefficients to ion species with varying sizes of ion.
This over-simplification of cPNP systems could cause a severe defect to model in many biological situations. 
 Many vital features of ion channels, such as
selectivity, rely on ion the sizes critically. For example, Na$^+$ (sodium) and K$^+$ (potassium), have the same valences (number of charges per particle); however, they have significantly distinctive biological characteristics due chiefly to their different sizes. That is the inconsistency in their ionic sizes that allows some channels to prefer Na$^+$ to K$^+$, and some to prefer K$^+$ to Na$^+$.

To consider ion sizes, one needs to include the excess chemical potential, beyond the ideal, in the model. To value the ion size effects, one may combine an element to the PNP-type model, which is a hard-sphere (HS) potential. The PNP system, combined with Density Functional Theory (DFT) for hard-sphere potentials of ion species, serves the purpose for this consideration and has been studied computationally with significant improvements \cite{GNE,JL12,SL18}.


 Including HS potential models of the excess electrochemical potential is the first step toward more reliable modeling and is required to account for ion size effects on the physiology of ionic flows. There are two classes of models for HS potentials, local and nonlocal. Local models for HS potentials, like the model defined in \eqref{assump} used in this paper, depend pointwise on ion concentrations, while nonlocal models are aimed as functionals of ion concentrations (see the DFT-PNP system  defined in \cite{JL12}). The PNP systems
with ionic sizes have been examined computationally for ion channels, and have recorded magnificent progress \cite{LLYZ13}.

One of the earliest local models for HS potentials was proposed by Bikerman \cite{B42}, which contains an ion size effect of mixtures but is not ion-specific. The HS potential is assumed to be the same for different ion species. Local models undoubtedly have the benefit of simplicity relative to nonlocal ones. 
In \cite{LLYZ13}, the authors study a one-dimensional version of a PNP-type system with a local model for the HS potential. The problem here in this work has fundamentally the same setting as in \cite{LLYZ13}; however, it includes nonzero permanent charges to the system, which makes it much more intricate. More models can be seen in \cite{HEL10}.

This paper is formed as follows. The local hardsphere PNP model for ionic 
flows with nonzero permanent charges is introduced in Section \ref{Sec.PNP} to set the platform for examinations in the following sections. Section \ref{Sec.GSPT} is the main body of this work where we apply the GSP theory on the PNP system to turn the BVP into a connecting system to acquire a nonlinear algebraic system of equations, called the matching system. In Section \ref{Sec.Ionsize-effects}, using the matching system and applying the zero-current condition, we derive zero-current fluxes that will lead to some interesting results. We then explore the qualitative impact of ion sizes to obtain some relevant results. In particular, we study the effects of ion size on zero-current fluxes. In Section \ref{Sec.Conclusion}, we express the importance and productivity of the work. We also bring up some future works we are working on and one may also take into consideration.

\section{Poisson-Nernst-Planck systems with a local HS model for ionic flows}\label{Sec.PNP}
\setcounter{equation}{0}

\subsection{A one-dimensional PNP-type system.} 
Since the ion channels have thin cross-sections compared to their lengths, three-dimensional PNP type models can be virtually viewed as one-dimensional models standardized over the interval [0, 1], where the interior and the exterior of the channel are bound. A genuine one-dimensional  PNP model for ionic flows of $n$ ion species is (\cite{EL07,LW10,NE98})

\begin{equation}\label{PNP}
\begin{array}{l}
  \dfrac{1}{h(x)} \dfrac{\partial}{\partial x}\Big( \varepsilon_r(x) \varepsilon_0 h(x) \dfrac{\partial \Phi}{\partial x} \Big) = - e\Big(\sum_{j=1}^n z_j c_j + Q(x) \Big), \\\\
 \dfrac{\partial c_k}{\partial t} +\dfrac{1}{h(x)} \dfrac{\partial \mathcal{J}_k}{\partial x} = 0, \quad -\mathcal{J}_k = \dfrac{1}{k_B T} D_k(x) h(x) c_k \dfrac{d\mu_k}{dx}, \quad k=1,2, ..., n,
  \end{array}
  \end{equation}

where $e$ is the elementary charge, $k_B$ is the Boltzmann constant, $T$ is the absolute temperature; $\Phi$ is the electric potential, $Q(x)$ is the permanent charge of the channel, $\varepsilon_r(x)$ is the relative
dielectric coefficient, $\varepsilon_0$ is the vacuum permittivity; $h(x)$ is the area of the cross section of the channel over the point $x$; and for the $k-$th ion species, $c_k$ is the concentration, $z_k$ is the valence (the number of charges per particle), $\mu_k$ is the electrochemical potential, $\mathcal{J}_k$ is the flux density, and $D_k(x)$ is the diffusion coefficient. The boundary conditions are, for $k = 1, 2, ..., n,$

\[
\Phi(t,0) = \mathcal{V}, \quad c_k(t,0)= l_k >0 ;  \quad \Phi(t,1) =0,  \quad c_k(t,1)=r_k >0.
\]

\medskip
\noindent {\bf \underline{Excess potential and a local HS model.}}
We discuss an extension of the geometric singular perturbation approach for classical PNP models to include ion size effects on ionic flow properties.
Many properties of electrodiffusion  are sensitive to ion sizes. Microscopic chemical properties do depend on ion sizes in addition to others -ionic radius is an essential factor in the periodic table of chemical elements.

A one-dimensional ion specific model for hard-sphere potential is
\begin{equation}\label{hsphere} 
\dfrac{1}{k_B T} \mu_k^{HS} = -\ln \Big(1- \sum_j d_j c_j \Big) + \dfrac{d_k \sum_j c_j}{1- \sum_j d_j c_j}
\end{equation}
where $d_j$ is the ionic diameter of $j$th ion species. The factor $d_k$ in the second term makes the model ion specific. This model for $\mu_k^{HS}$ is the local version of the nonlocal model of Percus-Yevick approximation....
For definiteness, we take the following setting:
\begin{itemize}
\item[(A1).] We consider two ion species ($n=2$) with $z_1 > 0> z_2$.
\item[(A2).] For the electrochemical potential $\mu_k$, in addition to the ideal component $\mu_k^{id}$, we also include the local hard-sphere potential $\mu_k^{HS}$ in \eqref{hsphere}.
\item[(A3).] The relative dielectric coefficient and the diffusion coefficient are constants, that is, $\varepsilon_r(x)= \varepsilon_r$ and $D_k(x)= D_k$.
\end{itemize}

\subsection{The steady-state boundary value problem and assumptions. }
Under the assumptions $(A1)-(A3)$, the steady-state system of PNP model in \eqref{PNP} is
\begin{equation}\label{assump}
\begin{array}{l}
  \dfrac{1}{h(x)} \dfrac{d}{dx}\Big( \varepsilon_r(x) \varepsilon_0 h(x) \dfrac{d\Phi}{dx} \Big) = - e(z_1c_1 + z_2 c_2)- Q(x), \\\\
  \dfrac{d\mathcal{J}_k}{dx} = 0, \quad -\mathcal{J}_k = \dfrac{1}{k_B T} D_k(x) h(x) c_k \dfrac{d\mu_k}{dx}, \quad k=1,2,
  \end{array}
  \end{equation}
with the following boundary conditions
\begin{equation}\label{BC}
\Phi(0) = \mathcal{V}, \quad c_k(0)= l_k >0 ;  \quad \Phi(1) =0,  \quad c_k(1)=r_k >0.
\end{equation}
We now make the dimensionless re-scaling in \eqref{assump},
$$
\phi = \dfrac{e}{k_B T} \Phi, \quad V= \dfrac{e}{k_B T} \mathcal{V}, \quad
\varepsilon^2= \dfrac{\varepsilon_r \varepsilon_0 k_B T}{e^2}, \quad J_k = \dfrac{\mathcal{J}_k}{D_k}.
$$
Then, for $k=1,2,$
\begin{align*}
-J_k&= - \dfrac{\mathcal{J}_k}{D_k} = \dfrac{1}{k_BT}h(x)c_k\dfrac{d\mu_k^{id}}{dx} + \dfrac{1}{k_BT} h(x)c_k \dfrac{d\mu_k^{HS}}{dx} \\
&= z_k h(x) c_k \dfrac{d\phi}{dx} + h(x) \dfrac{dc_k}{dx} + \dfrac{h(x)c_k}{k_B T}\dfrac{d\mu_k^{HS}}{dx}.
\end{align*}
Note also that,
$$
\varepsilon_r \varepsilon_0 \dfrac{d\Phi}{dx}= \varepsilon^2 \dfrac{e^2}{k_B T}\dfrac{d\Phi}{dx}= \varepsilon^2 \dfrac{e^2}{k_B T} \dfrac{k_B T}{e}\dfrac{d\phi}{dx}= \varepsilon^2 e \dfrac{d\phi}{dx}. 
$$
Therefore, system \eqref{assump} becomes
\begin{equation}\label{1dPNPdim}
 \begin{array}{l}
\dfrac{\varepsilon^2}{h(x)} \dfrac{d}{dx}\Big( h(x) \dfrac{d\phi}{dx} \Big) = -z_1c_1 - z_2 c_2- Q(x), \quad \dfrac{dJ_1}{dx}=\dfrac{dJ_2}{dx}= 0,\\\\
h(x) \dfrac{dc_1}{dx}+ z_1 h(x) c_1 \dfrac{d\phi}{dx} + \dfrac{h(x)c_1}{k_B T}\dfrac{d\mu_1^{HS}}{dx} = -J_1,\\\\
h(x) \dfrac{dc_2}{dx}+ z_2 h(x) c_2 \dfrac{d\phi}{dx} + \dfrac{h(x)c_2}{k_B T}\dfrac{d\mu_2^{HS}}{dx} = -J_2.
\end{array}
\end{equation} 
It follows directly from \eqref{hsphere} for the local hard-sphere potential $\mu_k^{HS}$ for two ion species that
\begin{equation}\label{hsphere2}
   \begin{array}{l}
\dfrac{1}{k_B T}\dfrac{d\mu_1^{HS}}{dx} = \dfrac{d_1(2+d_1(c_2-c_1) - 2d_2c_2)}{(1-d_1c_1 - d_2c_2)^2}\dfrac{dc_1}{dx}+ \dfrac{d_1+d_2-d_1^2c_1-d_2^2c_2}{(1-d_1c_1 - d_2c_2)^2}\dfrac{dc_2}{dx},\\\\
\dfrac{1}{k_B T}\dfrac{d\mu_2^{HS}}{dx} = \dfrac{d_1+d_2- d_1^2c_1-d_2^2c_2 }{(1-d_1c_1 - d_2c_2)^2}\dfrac{dc_1}{dx}+ \dfrac{d_2(2+d_2(c_1-c_2) - 2d_1c_1)}{(1-d_1c_1 - d_2c_2)^2}\dfrac{dc_2}{dx}.
\end{array}
\end{equation}
Substituting \eqref{hsphere2} into system \eqref{1dPNPdim}, we obtain
\begin{equation}\label{PNP2}
   \begin{array}{l}
 \dfrac{\varepsilon^2}{h(x)} \dfrac{d}{dx}\Big( h(x) \dfrac{d\phi}{dx} \Big) = -z_1c_1 - z_2 c_2- Q(x), \quad \dfrac{dJ_1}{dx}=\dfrac{dJ_2}{dx}= 0,\\
\dfrac{dc_1}{dx} = -f_1(c_1, c_2; d_1, d_2) \dfrac{d\phi}{dx} - \dfrac{1}{h(x)} g_1(c_1, c_2, J_1, J_2; d_1, d_2),\\
\dfrac{dc_2}{dx} = -f_2(c_1, c_2; d_1, d_2) \dfrac{d\phi}{dx} - \dfrac{1}{h(x)} g_2(c_1, c_2, J_1, J_2; d_1, d_2),
\end{array}
\end{equation}
with the boundary conditions,
\begin{equation}\label{BC2}
\phi(0) = V, \quad c_k(0)= l_k >0 ;  \quad \phi(1) =0,  \quad c_k(1)=r_k >0,
\end{equation}
where $f_j=f_j(c_1, c_2; d_1, d_2),~g_j=g_j(c_1, c_2; J_1, J_2; d_1, d_2)$ for $j=1,2$ are defined as follow,
$$
   \begin{array}{l}
f_1 =  z_1c_1 - (d_1 + d_2 - d_1^2c_1 - d_2^2c_2)(z_1c_1 + z_2 c_2)c_1-z_1(d_1 - d_2)c_1^2,\\
f_2 =  z_2c_2 - (d_1 + d_2 - d_1^2c_1 - d_2^2c_2)(z_1c_1 + z_2 c_2)c_2+z_2(d_1 - d_2)c_2^2,\\
g_1 =  ((1-d_1c_1)^2 + d_2^2 c_1c_2) J_1 - c_1(d_1+d_2 - d_1^2c_1 - d_2^2c_2)J_2,\\
g_2 =  ((1-d_2c_2)^2 + d_1^2 c_1c_2) J_2- c_2(d_1+d_2 - d_1^2c_1 - d_2^2c_2)J_1.
\end{array}
$$
Recall that $d_1$ and $d_2$ are the diameters of the two ion species. While $d_1 >0$ and $d_2>0$ are small, their ratio is of order $O(1)$. We thus set, for some $\lambda >0$,
$$
d_1=d, \quad \text{and} \quad d_2=\lambda d.
$$
Hence, $f_j, g_j$, for $j=1,2$, in above, become,
\begin{equation}\label{fg2}
   \begin{array}{l}
f_1(c_1, c_2, d) = z_1c_1 - \Big(2z_1c_1+ (1+\lambda)z_2c_2 \Big)c_1d + (c_1+\lambda^2c_2)(z_1c_1+z_2c_2) c_1d^2 ,\\
f_2(c_1, c_2, d) = z_2c_2 - \Big(2\lambda z_2c_2+ (1+\lambda)z_1c_1 \Big)c_2d + (c_1+\lambda^2c_2)(z_1c_1+z_2c_2)c_2d^2 ,\\
g_1(c_1, c_2, J_1, J_2, d) = J_1 - \Big( 2J_1+ (1+\lambda) J_2\Big) c_1d + (c_1+\lambda^2 c_2)(J_1+J_2) c_1d^2,\\
g_2(c_1, c_2, J_1, J_2, d) = J_2 - \Big( 2 \lambda J_2+ (1+\lambda) J_1\Big) c_2d + (c_1+\lambda^2 c_2)(J_1+J_2) c_2d^2.
\end{array}
\end{equation}
We assume the permanent charge $Q(x)$ is given by a piecewise constant function. The permanent charge $Q$ is piecewise constant  with one nonzero region; that is, for a partition $0<a<b<1$ of   $[0,1]$, 
 \begin{align}\label{Q}
 Q(x)=\left\{\begin{array}{ll}
 Q_1=Q_3=0, & x\in (0,a)\cup (b,1),\\
 Q_2, & x\in (a,b),
 \end{array}\right.
 \end{align}   
 where $Q_2$ is a constant.
It is easy to extend the work for $Q$ with multiple regions of nonzero constants.
For convenience, we set for $k=0,1$,
\begin{equation}\label{Abbrev}
\begin{aligned}
I_{k} :=& z_1J_{1k}+z_2J_{2k}, \quad T_{k} := J_{1k}+ J_{2k}, \quad \Lambda_k := J_{1k} + \lambda J_{2k}, \\
 \sigma(y) :=& (z_1-z_2)z_1c_{10}(y) - z_2Q, \quad H(x) := \int_0^x \dfrac{1}{h(s)}ds, \\
Z_0(y) :=&  \dfrac{z_1z_2 {I}_0 Q}{\sigma(y) }, \quad   w(a,b)= a+\lambda b + \frac{\lambda z_1 -z_2}{z_1 -z_2}(a+b).
\end{aligned}
\end{equation}


\noindent To end this section, we treat $\varepsilon > 0$ small as a singular perturbation parameter and $d_k's$ as regular perturbation parameters, and rewrite system \eqref{PNP2} into a standard form for singularly perturbed systems and convert the boundary value problem \eqref{PNP2} and \eqref{fg2} to a connected problem.\\
Denote derivative with respect to $x$ by overdot and introduce $u= \varepsilon \overset{.}{\phi}$
and $\tau =x$. System \eqref{PNP2} becomes
\begin{equation}  \label{Sloweps1}
\begin{aligned}
\varepsilon \overset{.}{\phi} &= u, \quad \varepsilon \overset{.}{u} = -z_1c_1 -z_2c_2 - Q(\tau)- \varepsilon \dfrac{h_{\tau}(\tau)}{h(\tau)}u,\\
\varepsilon \overset{.}{c}_1 &= -f_1 u - \dfrac{\varepsilon}{h(\tau)} g_1,\quad 
\varepsilon \overset{.}{c}_2 = -f_2 u - \dfrac{\varepsilon}{h(\tau)} g_2,\\
\overset{.}{J_1} &= \overset{.}{J_2} = 0, \quad \overset{.}{\tau}=1,
 \end{aligned}
\end{equation}
that is called the {\em slow system}.
 Let $B_l$ and $B_r$ be the subsets of the phase space $\mathbb{R}^7$ defined by

\[
\begin{array}{l}
B_l= \{(V,u,l_1, l_2, J_1, J_2, 0) \in \mathbb{R}^7 : \textrm{arbitrary}~u,J_1, J_2 \},\\
B_r= \{(0,u,r_1, r_2, J_1, J_2, 1) \in \mathbb{R}^7 : \textrm{arbitrary}~u,J_1, J_2 \},
\end{array}
\]
where $V$, $l_1$, $l_2$, $r_1$ and $r_2$ are given in \eqref{BC2}. Then the original boundary value problem is equivalent to a connecting problem, namely, finding a solution of \eqref{Sloweps1} from $B_l$ to $B_r$.
By setting $\varepsilon =0$ in system \eqref{Sloweps1}, we obtain the slow manifold,
$$
\mathscr{Z} = \Big\{ u=0,~ z_1c_1+ z_2c_2 + Q=0 \Big\}.
$$
For $\varepsilon > 0$, the rescaling $x=\varepsilon \xi$ of the independent variable $x$ give rise to the fast system, 
\begin{equation} \label{feps02}
\begin{array}{l}
\phi' = u, \quad u' = -z_1c_1 -z_2c_2 - Q - \varepsilon \dfrac{h_{\tau}(\tau)}{h(\tau)}u,\\
c_1' = -f_1u - \dfrac{\varepsilon}{h(\tau)} g_1, \quad 
c_2' = -f_2u - \dfrac{\varepsilon}{h(\tau)} g_2,\\
J_1' = J_2' = 0, \quad \tau'=\varepsilon,
\end{array}
\end{equation}
where prime denotes the derivative with respect to the fast variable $\xi$. The limiting fast system is,
\begin{equation}\label{fzero02}
\begin{array}{l}
\phi' = u, \quad u' = -z_1c_1 -z_2c_2-Q,\\
c_1' = -f_1 u,\quad c_2' = -f_2 u,\\
J_1' = J_2' = 0, \quad \tau'=0.
\end{array}
\end{equation}

The following Lemma can be directly verified.
\begin{lem}\label{0nh}
The slow manifold $\mathscr{Z}$ is normally hyperbolic for any permanent charge $Q$.
\end{lem}



\section{GSP Theory for the connecting problem.}\label{Sec.GSPT}
\setcounter{equation}{0}

We apply the general geometric singular perturbation Theory (GSP) to construct singular orbits for the connecting problem \cite{Hek,Jones95,JK94,Kuehn15,Liu00}.
We will first construct singular orbits on each sub-interval $[x_{j-1}, x_j]$ where $Q(x)$ is constant and then match them at jump points $x=x_j's$ of $Q(x)$. To do so, we pre-assign the values of $\phi,~c_k's$ at $x_j$ for $j=1,2$,
\[
\phi(x_j)= \phi^{[j]},\quad c_1(x_j)=c_1^{[j]}, \quad c_2(x_j)=c_2^{[j]},
\]
with given $\phi^{[0]}=V$ and $c_1^{[0]}=l_1, \quad c_2^{[0]}=l_2$ at $x=0$, and $\phi^{[3]}=0$ and $c_1^{[3]}=r_1, \quad c_2^{[3]}=r_2$ at $x=1$.
Note that we have introduced $6$ unknown variables. For $j=0,1,2,3,$ introduce the sets
$$
B_j= \{ (\phi, u, c_1, c_2, J, \tau)~:~~\phi=\phi^{[j]},~ c_1=c_1^{[j]},~ c_2=c_2^{[j]}, ~\tau = x_j \}.
$$
Note that $B_0=B_l$ and $B_3=B_r$. The next step is to construct singular orbits over each interval $[x_{j-1}, x_j]$ for the connecting problem between $B_{j-1}$ and $B_j$. Finally, we match the singular orbits at each $x_j$ to obtain singular orbits over the whole interval $[0,1]$. For a singular orbit on the whole interval $[0,1]$, we require that
$$
J_1^{l}=J_1^{m}=J_1^{r}, \quad J_2^{l}=J_2^{m}=J_2^{r}, \quad u^{a,l}=u^{a,m}, \quad u^{b,m}=u^{b,r}.
$$
This consists of six conditions. The number of conditions is exactly the same as the number of unknown values in the above preassigned values. 


\subsection{A singular orbit on $[0,a]$ where $Q(x)=0$.}\label{sec-sing-0a}

Here we construct singular orbits for the connecting problem from $B_{0}$ to $B_1$. Each such an orbit will consist of two boundary layers $\Gamma^{l}$ at $x=0,~\Gamma^{a,l}$ at $x=a$, and a regular layer $\Lambda_1$ over the interval $[0,a]$.
\medskip

\noindent {\bf \underline{Dynamics and Boundary/Internal Layers on $[0,a]$.}}  By setting $\varepsilon =0$ in system \eqref{Sloweps1} with $Q=0$, the slow manifold is,
$
\mathscr{Z}_1 = \{ u=0,~ z_1c_1+ z_2c_2 =0 \}.
$
For $\varepsilon > 0$, the rescaling $x=\varepsilon \xi$ of the independent variable $x$ give rise to
\begin{equation} \label{feps01}
\begin{array}{l}
\phi' = u, \quad u' = -z_1c_1 -z_2c_2 - \varepsilon \dfrac{h_{\tau}(\tau)}{h(\tau)}u,\\
c_1' = -f_1 u - \dfrac{\varepsilon}{h(\tau)} g_1, \quad
c_2' = -f_2 u - \dfrac{\varepsilon}{h(\tau)} g_2,\\
J_1' = J_2' = 0, \quad \tau'=\varepsilon.
\end{array}
\end{equation}
where prime denotes the derivative with respect to the variable $\xi$. The limiting fast system is,
\begin{equation}\label{fzero01}
\begin{array}{l}
\phi' = u, \quad u' = -z_1c_1 -z_2c_2,\\
c_1' = -f_1u, \quad
c_2' = -f_2 u,\\
J_1' = J_2' = 0, \quad \tau'=0.
\end{array}
\end{equation}
The set of equilibria of \eqref{fzero01} is precisely $\mathscr{Z}_1$. From lemma \eqref{0nh}, when $Q=Q_1=0$, the slow manifold $\mathscr{Z}_1$ is {normally hyperbolic} for system \eqref{fzero01}.

We denote the stable \big(resp. unstable\big) manifold of $\mathscr{Z}_1$ by $W^s(\mathscr{Z}_1)$ \big(resp. $W^u(\mathscr{Z}_1)$\big). Let $M^{l}$ be the collection of orbits from $B_{0}$ in forward time under the flow of system \eqref{feps01} and $M^{a,l}$ be the collection of orbits from $B_1$ in backward time under the flow of system \eqref{feps01}. Then, for a singular orbit connecting $B_{0}$ to $B_1$, the boundary layer at $\tau =x = 0$ must lie in $N^{l} = M^{l} \cap W^s(\mathscr{Z}_1)$ and the boundary layer at $\tau = x =a$ must lie in $N^{a,l} = M^{a,l} \cap W^u(\mathscr{Z}_1)$.  
We look for solutions
$$
\Gamma(\xi;d)= \big(\phi(\xi;d), u(\xi;d), c_1(\xi;d), c_2(\xi;d), J_1(d), J_2(d), \tau\big)
$$
of system \eqref{fzero01} of the form
\begin{equation}\label{dsys}
\begin{array}{l}
\phi(\xi;d) = \phi_0(\xi) + \phi_1(\xi) d + o(d), \quad u(\xi;d) = u_0(\xi) + u_1(\xi) d + o(d),\\
c_1(\xi;d) = c_{10}(\xi) + c_{11}(\xi) d + o(d), \quad c_2(\xi;d) = c_{20}(\xi) + c_{21}(\xi) d + o(d),\\
J_1(d)=J_{10}+J_{11}d+o(d), \quad J_2(d)=J_{20}+J_{21}d+o(d).
\end{array}
\end{equation}
Substituting \eqref{dsys} into system \eqref{fzero01}, we obtain, for the zeroth order in $d$,
\begin{equation}\label{f0th0a}
\begin{array}{l}
\phi_0'=u_0, \quad u_0'= -z_1c_{10}-z_2c_{20},\\
c_{10}'= -z_1c_{10}u_0,\quad c_{20}'= -z_2c_{20}u_0,\\
J_{10}'=J_{20}'=0, \quad \tau'=0,
\end{array}
\end{equation}
and for the first order in $d$,
\begin{equation}\label{f1st0a}
\begin{array}{l}
\phi_1'=u_1, \quad u_1'= -z_1c_{11}-z_2c_{21},\\
c_{11}'= -z_1c_{11}u_0  -z_1c_{10}u_1 + \Big(2z_1 c_{10} + (1+ \lambda)z_2 c_{20}  \Big)c_{10}u_0,\\
c_{21}'= -z_2c_{21}u_0  -z_2c_{20}u_1 + \Big((1+\lambda)z_1 c_{10} + 2\lambda z_2 c_{20}\Big)c_{20}u_0,\\
J_{11}'=J_{21}'=0.
\end{array}
\end{equation}


\begin{lem}\label{lem0th1}
The zeroth order system \eqref{f0th0a} has a complete set of first integrals, 
$$
\begin{array}{l}
\mathcal{H}_{10,l} = e^{z_1 \phi_0} c_{10}, \quad \mathcal{H}_{20,l} = e^{z_2 \phi_0} c_{20}, \quad \mathcal{H}_{30,l}=J_{10}, \quad \mathcal{H}_{40,l}=J_{20},\\
 \mathcal{H}_{50,l}=c_{10}+c_{20}-\dfrac{1}{2}u_0^2, \quad \mathcal{H}_{60,l}=\tau,
\end{array}
$$
and the first order system \eqref{f1st0a} has a complete set of first integrals,
$$
\begin{array}{l}
\mathcal{H}_{11,l} = z_1 \phi_1 + c_{11}/c_{10} + 2c_{10} + (\lambda+1)c_{20},\quad
\mathcal{H}_{21,l} = z_2 \phi_1 +c_{21}/c_{20}+ 2\lambda c_{20} + (\lambda+1)c_{10},\\
\mathcal{H}_{31,l}=u_0u_1 - c_{11} -c_{21} - (\lambda+1) c_{10}c_{20} - c_{10}^2 - \lambda c_{20}^2  ,\quad
\mathcal{H}_{41,l}=J_{11}, \quad \mathcal{H}_{51,l}=J_{21}.
\end{array}
$$
\begin{proof}
It can be verified directly from \eqref{f0th0a} and \eqref{f1st0a}. 
\end{proof}
\end{lem}
Recall that we are interested in the solutions of $\Gamma^{l}(\xi;d) \subset N^{l} = M^{l} \cap W^s(\mathscr{Z}_1)$ with $\Gamma^{l}(0;d) \in B_{0}$, and $\Gamma^{a,l}(\xi;d) \subset N^{a,l} = M^{a,l} \cap W^u(\mathscr{Z}_1)$ with $\Gamma^{a,l}(0;d) \in B_{1}$.
\begin{prop}\label{propHS0a}
Assume that $d\geq 0$ is small.\\
(i) The stable manifold $W^s(\mathscr{Z}_1)$ intersects $B_0$ transversally at points
$$
\big(V, u_0^{l}+ u_1^{l}d+ o(d), l_1, l_2, J_1(d), J_2(d), 0\big),
$$
and the $\omega-$limit set of $N^{l}=M^{l}\cap W^s(\mathscr{Z}_1)$ is
$$
\omega\big(N^{l}\big)=\Big\{ \big( \phi_0^{l}+\phi_1^{l}d+o(d), 0, c_{10}^{l}+c_{11}^{l}d+o(d), c_{20}^{l}+c_{21}^{l}d+o(d), J_1(d), J_2(d), 0\big)\Big\},
$$
where $J_k(d)=J_{k0}+J_{k1}d+o(d),~~k=1,2,$ can be arbitrary and
\[
\begin{array}{l}
\phi_0^{l} = V - \frac{1}{z_1 -z_2}\ln \frac{-z_2 l_2}{z_1l_1}, \quad z_1c_{10}^{l}=-z_2c_{20}^{l}=\big(z_1l_1 \big)^{\frac{-z_2}{z_1-z_2}}\big(-z_2l_2 \big)^{\frac{z_1}{z_1-z_2}},\\
u_0^{l}= sgn(z_1l_1+z_2l_2)\sqrt{2\Big(l_1+l_2+ \frac{z_1 -z_2}{z_1z_2} \big(z_1l_1 \big)^{\frac{-z_2}{z_1-z_2}}\big(-z_2l_2 \big)^{\frac{z_1}{z_1-z_2}} \Big)},\\
\phi_1^{l}= \frac{1-\lambda}{z_1 -z_2} \big(l_1+l_2-c_{10}^{l}-c_{20}^{l} \big),\\
z_1c_{11}^{l}=-z_2c_{21}^{l}=z_1c_{10}^{l}\Big(w(l_1, l_2) + \frac{2(\lambda z_1 -z_2)}{z_2}c_{10}^{l} \Big),\\
u_1^{l} =\dfrac{1}{u_0^{l}}\Big( (l_1+l_2)(l_1+\lambda l_2) - (c_{10}^{l}+c_{20}^{l})(c_{10}^{l}+\lambda c_{20}^{l}) - c_{11}^{l} -c_{21}^{l} \Big),
\end{array}
\]
where, $w(l_1, l_2) $ was defined in \ref{Abbrev}.\\
(ii) The unstable manifold $W^u(\mathscr{Z}_1)$ intersects $B_1$ transversally at points
$$
\big(\phi_0^a + \phi_1^a d + o(d), u_0^{a,l}+ u_1^{a,l}d+ o(d), c_{10}^{a} + c_{11}^{a}d +o(d), c_{20}^{a}  + c_{21}^{a}d +o(d), J_1(d), J_2(d), a\big),
$$
and the $\alpha-$limit set of $N^{a,l}=M^{a,l}\cap W^u(\mathscr{Z}_1)$ is
$$
\alpha\big(N^{a,l}\big)=\Big\{ \big( \phi_0^{a,l}+\phi_1^{a,l}d+o(d), 0, c_{10}^{a,l}+c_{11}^{a,l}d+o(d), c_{20}^{a,l}+c_{21}^{a,l}d+o(d), J_1(d), J_2(d), a\big)\Big\},
$$
where $J_k(d)=J_{k0}+J_{k1}d+o(d),~~k=1,2,$ can be arbitrary and
\begin{equation}\label{prop1:22}
\begin{array}{l}
\phi_0^{a,l} = \phi_0^{a} - \frac{1}{z_1 -z_2}\ln \dfrac{-z_2 c_{20}^{a}}{z_1c_{10}^{a}}, \quad z_1c_{10}^{a,l}=-z_2c_{20}^{a,l}=\big(z_1c_{10}^{a} \big)^{\frac{-z_2}{z_1-z_2}}\big(-z_2c_{20}^{a} \big)^{\frac{z_1}{z_1-z_2}},\\
u_0^{a,l}= -sgn(z_1c_{10}^{a}+z_2c_{20}^{a})\sqrt{2\Big(c_{10}^{a}+c_{20}^{a}+ \frac{z_1 -z_2}{z_1z_2} \big(z_1c_{10}^{a} \big)^{\frac{-z_2}{z_1-z_2}}\big(-z_2c_{20}^{a} \big)^{\frac{z_1}{z_1-z_2}} \Big)},\\
\phi_1^{a,l} = \phi_1^{a} + \dfrac{1}{z_1-z_2  }\Big(  \dfrac{c_{11}^{a}}{c_{10}^{a}} - \dfrac{c_{21}^{a}}{c_{20}^{a}} +(1-\lambda )\big(c_{10}^{a}+c_{20}^{a} -c_{10}^{a,l} -c_{20}^{a,l} \big) \Big),\\
z_1c_{11}^{a,l}=-z_2c_{21}^{a,l}=z_1c_{10}^{a,l}\Big(\frac{1}{z_1-z_2}\big( z_1 \frac{c_{21}^{a}}{c_{20}^{a}} -z_2 \frac{c_{11}^{a}}{c_{10}^{a}} \big) + w (c_{10}^{a}, c_{20}^{a}) + \frac{2(\lambda z_1 -z_2)}{z_2}c_{10}^{a,l} \Big),\\
u_1^{a,l}= \dfrac{1}{u_0^{a,l}}\Big( ( c_{10}^{a}+ c_{20}^{a})(c_{10}^{a}+ \lambda  c_{20}^{a}) - ( c_{10}^{a,l}+ c_{20}^{a,l})(c_{10}^{a,l}+ \lambda  c_{20}^{a,l})\\
\qquad \qquad \qquad +  c_{11}^{a}+c_{21}^{a} - c_{11}^{a,l}- c_{21}^{a,l} \Big).
\end{array}
\end{equation}
\end{prop}


\begin{proof}
The stated result for the zeroth order system has been obtained in \cite{EL07, Liu05, Liu09}. For the first order system, from first integrals in Lemma \eqref{lem0th1} we establish the results for $\phi_1^{l}, c_{11}^{l}, c_{21}^{l}$ and $u_1^{l}$ for system \eqref{f1st0a}. We emphasize that those for  $\phi_1^{a,l}, c_{11}^{a,l}, c_{21}^{a,l}$ and $u_1^{a,l}$ are different here because we have extra terms $c_{11}^{a},~c_{21}^{a}$ and $\phi_{1}^{a}$. 
One should also note that $\phi_1(0)=c_{11}(0)=c_{21}(0)=0$. Using the first integrals $H_{11,l}$ and $H_{21,l}$, since $H_{j1,l}(\xi)=H_{j1,l}(0),~~j=1,2$, then
taking the limit as $\xi \to \infty$, we have
$$
\begin{aligned}
&c_{11}^{l}= c_{10}^{l} \Big( 2l_1+ (\lambda +1) l_2 - 2 c_{10}^{l} - (\lambda+ 1) c_{20}^{l} - z_1 \phi_1^{l}\Big),\\
&c_{21}^{l}= c_{20}^{l} \Big( 2\lambda l_2+ (\lambda +1) l_1 - 2\lambda c_{20}^{l} - (\lambda+ 1) c_{10}^{l} - z_2 \phi_1^{l}\Big).
\end{aligned}
$$
When $\xi \to \infty$, over the slow manifold $\mathscr{Z}_1$, where $z_1c_1+z_2c_2=0$, it follows from  \eqref{dsys} that $z_1c_{10}^{l}+z_2c_{20}^{l}=z_1c_{11}^{l}+z_2c_{21}^{l}=0$. Hence, from two above equations, $\phi_1^{l}$ will be obtained.
Then, $c_{11}^{l}$ and $c_{21}^{l}$ follow directly.
To derive the formula for $u_1^{l} = u_1(0)$, in view of $H_{31,l}(0) = H_{31,l}(\infty)$, we have
$$
u_0^{l}u_1^{l} - (\lambda +1 ) l_1l_2 -l_1^2 - \lambda l_2^2 = -c_{11}^{l} - c_{21}^{l} -(\lambda +1) c_{10}^{l}c_{20}^{l} -(c_{10}^{l})^2 - \lambda (c_{20}^{l})^2.
$$
The formula for $u_1^{l}$ follows directly.

\noindent Similarly, we find $c_{11}^{a,l}, c_{21}^{a,l}, \phi_{1}^{a,l}$ and $u_{1}^{a,l}$. 
Using the first integrals $H_{11,l}$ and $H_{21,l}$ in Lemma \ref{lem0th1}, since $H_{j1,l}(\xi)=H_{j1,l}(\infty)$ for $j=1,2$, taking the limit as $\xi \to 0$ one has,
$$
\begin{aligned}
&c_{11}^{a,l}= c_{10}^{a,l} \Big(z_1\phi_1^{a}+\frac{c_{11}^{a}}{c_{10}^{a}}+ 2c_{10}^{a}+ (\lambda +1) c_{20}^{a} - 2 c_{10}^{a,l} - (\lambda+ 1) c_{20}^{a,l} - z_1 \phi_1^{a,l}\Big),\\
&c_{21}^{a,l}= c_{20}^{a,l} \Big(z_2\phi_1^{a}+\frac{c_{21}^{a}}{c_{20}^{a}}+ 2\lambda c_{20}^{a}+ (\lambda +1) c_{10}^{a} - 2\lambda c_{20}^{a,l} - (\lambda+ 1) c_{10}^{a,l} - z_2 \phi_1^{a,l}\Big).
\end{aligned}
$$
Note that when $\xi \to 0$, we are on the slow manifold $\mathscr{Z}_1$ where $z_1c_1+z_2c_2=0$; hence \eqref{dsys} implies $z_1c_{10}^{a,l}+z_2c_{20}^{a,l}=z_1c_{11}^{a,l}+z_2c_{21}^{a,l}=0$. 
Then the formula for $\phi_1^{a,l}$ will be obtained.
To find $c_{11}^{a,l}$, we just need to substitute $\phi_1^{a,l}$ into the equation for $c_{11}^{a,l}$ in above.
We can find the formula for $u_1^{a,l} = u_1(\infty)$, directly from $H_{31,l}(0) = H_{31,l}(\infty)$.
\end{proof}

\noindent {\bf \underline{Limiting slow dynamics and regular layers on $[0,a]$.}}
Now we construct the regular layer on $\mathscr{Z}_1$ that connects $\omega(N^{l})$ and $\alpha(N^{a,l})$. Note that, for $\varepsilon =0$, system \eqref{Sloweps1} with $Q=0$ loses most information. To fix this problem, we follow the idea in \cite{EL07, Liu05, Liu09} and make a re-scaling $u=\varepsilon p$ and $-z_2c_2=z_1c_1+ \varepsilon q$ in system \eqref{Sloweps1}. In terms of the new variables, system \eqref{Sloweps1} with $Q=0$ becomes
\[
\begin{array}{l}
\overset{.}{\phi} = p, \quad \varepsilon\overset{.}{p} = q- \varepsilon \dfrac{h_{\tau}(\tau)}{h(\tau)}p, \quad \varepsilon \overset{.}{q}= (z_1f_1 +z_2f_2) p + \dfrac{z_1g_1+z_2g_2}{h(\tau)}  \\
 \overset{.}{c}_1 = -f_1p - \dfrac{g_1}{h(\tau)}, \quad \overset{.}{J_1} = \overset{.}{J_2} = 0, \quad \overset{.}{\tau}=1,
\end{array}
\]
where, for $k=1,2,$
$$
f_k=f_k\Big( c_1, -\dfrac{z_1c_1+\varepsilon q}{z_2};d, \lambda d \Big)\quad \text{and}\quad 
g_k=g_k\Big( c_1, -\dfrac{z_1c_1+\varepsilon q}{z_2}, J_1, J_2;d, \lambda d \Big),
$$
were defined in \eqref{fg2}.
This system is also a singular perturbation problem. Its limiting slow system is,
\begin{equation}\label{1limslow}
\begin{array}{l}
q = 0, \quad p = - \dfrac{\sum_{j=1}^2 z_jg_j\big(c_1,-\frac{z_1}{z_2}c_1, J_1, J_2; d, \lambda d \big)}{z_1(z_1-z_2)h(\tau)c_1}, \quad
\overset{.}{\phi}=p,\\
\overset{.}c_1= -f_1\big(c_1,-\frac{z_1}{z_2}c_1; d, \lambda d \big)p - \dfrac{1}{h(\tau)} g_1\big(c_1,-\frac{z_1}{z_2}c_1, J_1, J_2; d, \lambda d \big),\\
\overset{.}{J_1} = \overset{.}{J_2} = 0, \quad \overset{.}{\tau}=1.
\end{array}
\end{equation}
Note that to get the expression to the bottom of $p$, we have used \eqref{fg2}, that is,
$$
z_1f_1\Big(c_1, -\frac{z_1}{z_2}c_1; d, \lambda d \Big) + z_2 f_2\Big(c_1, -\frac{z_1}{z_2}c_1; d, \lambda d \Big)= z_1\big(z_1- z_2 \big)c_1.
$$
From system \eqref{1limslow}, the slow manifold is
$$
\mathcal{S}= \Big\{ q=0,~p= -\dfrac{\sum_{j=1}^2 z_jg_j\big(c_1,-\frac{z_1}{z_2}c_1, J_1, J_2; d, \lambda d \big)}{z_1(z_1 -z_2) h(\tau) c_1} \Big\}.
$$
Therefore, the limiting slow system on $\mathcal{S}$ is
\begin{equation}\label{1limslow:SM}
\begin{array}{l}
\overset{.}{\phi}=p,\\
\overset{.}c_1= -f_1\big(c_1,-\frac{z_1}{z_2}c_1; d, \lambda d \big)p - \dfrac{1}{h(\tau)} g_1\big(c_1,-\frac{z_1}{z_2}c_1, J_1, J_2; d, \lambda d \big),\\
\overset{.}{J_1} = \overset{.}{J_2} = 0, \quad \overset{.}{\tau}=1.
\end{array}
\end{equation}
As for the layer problem, we look for solutions of \eqref{1limslow:SM} of the form
\begin{equation}\label{1slow:dsys}
\begin{array}{l}
\phi(x) = \phi_0(x) + \phi_1(x)d + o(d),\\
c_1(x) = c_{10}(x) + c_{11}(x)d + o(d),\\
J_1= J_{10}+J_{11}d+o(d), \quad J_2= J_{20}+J_{21}d+o(d),
\end{array}
\end{equation}
to connect $\omega\big(N^{l} \big)$ and $\alpha\big(N^{a,l} \big)$ given in Proposition \eqref{propHS0a}. In particular, for $j=0,1,$
$$
\big( \phi_j(0),c_{1j}(0)\big) = \big(\phi_j^{l}, c_{1j}^{l} \big), \quad \big( \phi_j(a),c_{1j}(a)\big) = \big(\phi_j^{a,l}, c_{1j}^{a,l} \big).
$$
From  \eqref{fg2}  over the slow manifold $\mathscr{Z}_1$, one obtains
$$
\begin{aligned}
&f_1\big( c_1, -\dfrac{z_1c_1}{z_2};d, \lambda d \big)=  z_1c_{10} + \Big(z_1 c_{11} + (\lambda - 1) z_1c_{10}^2\Big) d + o(d),\\
&g_1\big( c_1, -\dfrac{z_1c_1}{z_2};d, \lambda d \big)= J_{10} + \Big(  J_{11} -c_{10} \big[T_0 +\Lambda_0 \big] \Big)d + o(d),\\
&z_1g_1+z_2g_2 = I_0 +\Big( I_1 - (1-\lambda) z_1c_{10} T_0 \Big) d +o(d).
\end{aligned}
$$
Thus, It follows from above and \eqref{1limslow:SM} that,
\begin{equation}\label{1slow:27}
\begin{array}{l}
\overset{.}{\phi_0}=  \dfrac{-I_0}{z_1(z_1-z_2) h(\tau) c_{10}}, \quad \overset{.}{c}_{10}=  \dfrac{z_2T_0}{(z_1-z_2) h(\tau)},\\
\overset{.}{J}_{10} = \overset{.}{J}_{20} = 0, \quad \overset{.}{\tau} =1,
\end{array}
\end{equation}
and
\begin{equation}\label{1slow:28}
\begin{array}{l}
\overset{.}{\phi_1}=  \dfrac{I_0 c_{11}}{z_1(z_1-z_2) h(\tau) c_{10}^2} + 
\dfrac{(1-\lambda)z_1T_0c_{10} - I_1}{z_1(z_1-z_2) h(\tau) c_{10}},\\ 
\overset{.}{c}_{11}=  \dfrac{2(\lambda z_1 - z_2)T_0c_{10} + z_2T_1}{(z_1-z_2) h(\tau)}, \quad
\overset{.}{J}_{11} = \overset{.}{J}_{21} = 0.
\end{array}
\end{equation}


\begin{lem}\label{lem22}
For $0 \leq x \leq a$, there is a unique solution $(\phi_0(x), c_{10}(x), J_{10}, J_{20}, \tau(x)\big)$ of \eqref{1slow:27} such that 
\[
\big( \phi_0(0), c_{10}(0), \tau(0) \big) = \big( \phi_0^{l}, c_{10}^{l}, 0 \big) ~~\textrm{and}~~\big( \phi_0(a), c_{10}(a), \tau(a) \big) = \big( \phi_0^{a,l}, c_{10}^{a,l}, a \big) 
\]
where $\phi_0^{l},~\phi_0^{a,l},~c_{10}^{l}$, and $c_{10}^{a,l}$ are given in Proposition \eqref{propHS0a}. The solution, for $0 \leq x \leq a$, is given by
\begin{equation}\label{L2:31}
\begin{aligned}
&\phi_0(x) = \phi_0^{l}+ \dfrac{\phi_0^{a,l} - \phi_0^{l}}{\ln c_{10}^{a,l}- \ln c_{10}^{l}} \ln \Big(1- \dfrac{H(x)}{H(a)} + \dfrac{H(x)}{H(a)}\dfrac{c_{10}^{a,l}}{c_{10}^{l}} \Big), \\
&c_{10}(x) = \Big(1- \dfrac{H(x)}{H(a)} \Big)c_{10}^{l} + \dfrac{H(x)}{H(a)}c_{10}^{a,l},\\
&J_{10} = \dfrac{c_{10}^{l}- c_{10}^{a,l}}{H(a)} \Big(1+ \dfrac{z_1\big( \phi_0^{l}- \phi_0^{a,l}\big)}{\ln c_{10}^{l} - \ln c_{10}^{a,l}} \Big),\\
&J_{20}= -\dfrac{z_1\big(c_{10}^{l}- c_{10}^{a,l}\big)}{z_2H(a)} \Big(1+ \dfrac{z_2\big( \phi_0^{l}- \phi_0^{a,l}\big)}{\ln c_{10}^{l} - \ln c_{10}^{a,l}} \Big),\quad \tau(x)= x.
\end{aligned}
\end{equation}
\end{lem}
\begin{proof}
The solution of system \eqref{1slow:27} with the initial condition $\big(\phi_0^{l}, c_{10}^{l}, J_{10}, J_{20}, 0\big)$, for $0 \leq x \leq a$, is,
\[
\phi_0(x) = \phi_0^{l}- \frac{I_0}{z_1(z_1 -z_2)} \int_0^x \frac{1}{h(s)c_{10}(s)} ds, \quad 
c_{10}(x) = c_{10}^{l} +\frac{z_2T_0}{z_1 - z_2} H(x), \quad \tau(x)=x.
\]
It follows from the  second equation in above and at $x=a$ where $c_{10}(a) =c_{10}^{a,l}$ that
\begin{equation}\label{1sT0}
T_0 =  \dfrac{(z_1 - z_2)(c_{10}^{a,l} - c_{10}^{l})}{z_2H(a)}.
\end{equation}
We will find the above integral in \eqref{1f0}.
Thus, $
\phi_0(x) = \phi_0^{l}-  \frac{I_0H(a)\big( \ln c_{10}(x) - \ln c_{10}^{l} \big)}{z_1(z_1-z_2)\big(c_{10}^{a,l}-c_{10}^l\big)}.
$
Now, applying the boundary conditions $c_{10}(a) =c_{10}^{a,l}$ and $\phi_0(a) = \phi_{0}^{a,l}$ we have 
\begin{equation}\label{1sI0}
\begin{array}{l}
I_0 =  \dfrac{z_1(z_1 - z_2)(c_{10}^{l} - c_{10}^{a,l})(\phi_{0}^{l} - \phi_{0}^{a,l})}{H(a)(\ln c_{10}^{l}- \ln c_{10}^{a,l})}.
\end{array}
\end{equation}
The expressions for $J_{10}$ and $J_{20}$, and hence, for $\phi_0(x)$ and $c_{10}(x)$ follow directly.
\end{proof}

To find the first order terms and for convenience we first go through the following integrals, using equations \eqref{1slow:27}, \eqref{1sT0} and \eqref{1sI0}:
\begin{equation}\label{1f0}
\begin{aligned}
&\int_0^x \dfrac{1}{h(s)c_{10}(s)}ds = H(a)\dfrac{\ln c_{10}^{l}- \ln c_{10}(x)}{c_{10}^{l} - c_{10}^{a,l}},\quad \int_0^x \dfrac{c_{10}(s)}{h(s)}ds = \dfrac{H(a) \Big((c_{10}^{l})^2 - c_{10}^2(x)\Big)}{2\big(c_{10}^{l}- c_{10}^{a,l}\big)},\\
&\int_0^x \dfrac{1}{h(s)c_{10}^2(s)}ds = \dfrac{H(a)(c_{10}^{l} - c_{10}(x))}{c_{10}^{l}c_{10}(x)(c_{10}^{l}- c_{10}^{a,l})},\\
&\int_0^x \dfrac{H(s)}{h(s)c_{10}^2(s)}ds = \dfrac{H(a)}{(c_{10}^{l}- c_{10}^{a,l})} \Big( \dfrac{H(x)}{c_{10}(x)} - H(a)\dfrac{\ln c_{10}^{l} - \ln c_{10}(x)}{c_{10}^{l} - c_{10}^{a,l}}\Big).
\end{aligned}
\end{equation}
We also define the following three functions for convenience (Note that $w$ is defined in \eqref{Abbrev}),
\begin{equation}\label{MB1-MNP}
\begin{aligned}
M_{0} :=& c_{11}^{l}-c_{11}^{a,l}+  \dfrac{\lambda z_1 - z_2}{z_2}\big( (c_{10}^{a,l})^2 - (c_{10}^{l})^2\big)\\ 
=& c_{10}^{l} w(l_1,l_2) - c_{10}^{a,l} w(c_{10}^{a},c_{20}^{a}) + \dfrac{\lambda z_1 - z_2}{z_2}\big( (c_{10}^{l})^2 - (c_{10}^{a,l})^2\big) - c_{10}^{a,l}\dfrac{z_1\frac{c_{21}^{a}}{c_{20}^{a}} - z_2\frac{c_{11}^{a}}{c_{10}^{a}}}{z_1-z_2} , \\
N_{0}:=& \dfrac{(c_{10}^{l}- c_{10}^{a,l})}{\ln c_{10}^{l}- \ln c_{10}^{a,l}} \Bigg\{
\big( \phi_{1}^{l} - \phi_{1}^{a,l} \big) - \dfrac{1-\lambda }{z_2}\big(c_{10}^{l} - c_{10}^{a,l} \big)  +\dfrac{\phi_{0}^{l} - \phi_{0}^{a,l}}{c_{10}^{l}- c_{10}^{a,l}}M_{0}\\
&\qquad \qquad - \dfrac{w(l_1, l_2) -w(c_{10}^{a}, c_{20}^{a}) }{\big( \ln c_{10}^{l} - \ln c_{10}^{a,l} \big)}\big(\phi_{0}^{l} - \phi_{0}^{a,l} \big)+ \dfrac{\phi_{0}^{l} - \phi_{0}^{a,l}}{\ln c_{10}^{l}- \ln c_{10}^{a,l}}\dfrac{z_1\frac{c_{21}^{a}}{c_{20}^{a}} - z_2\frac{c_{11}^{a}}{c_{10}^{a}}}{z_1-z_2}\Bigg\},\\
P_0(x):=& \Big( w(l_1, l_2) + \dfrac{\lambda z_1 -z_2}{z_2}c_{10}^{l}\Big)\dfrac{\big(c_{10}^{l} - c_{10}(x)\big)}{c_{10}(x)}+ \dfrac{\lambda z_1 -z_2}{z_2}\dfrac{(c_{10}^{l} - c_{10}^{a,l})H(x)}{H(a)} \\
 &- \dfrac{M_{0}H(x)}{H(a)c_{10}(x)} +\dfrac{M_{0}\big(\ln c_{10}^{l} - \ln c_{10}(x)\big)}{c_{10}^{l} - c_{10}^{a,l}}.
\end{aligned}
\end{equation}

\begin{lem}\label{lem:1slow1st} 
For $0 \leq x \leq a$, there is a unique solution $\big(\phi_1(x), c_{11}(x), J_{11}, J_{21}, \tau(x) \big)$ of \eqref{1slow:28} such that
$$
\big(\phi_1(0), c_{11}(0), \tau(0) \big) = (\phi_1^{l}, c_{11}^{l}, 0)~~ \textrm{~and~}~~\big(\phi_1(a), c_{11}(a), \tau(a) \big) = (\phi	_1^{a,l}, c_{11}^{a,l}, a),
$$
where $\phi_1^{l}, \phi_1^{a,l}, c_{11}^{l}$ and $c_{11}^{a,l}$ are given in Proposition \eqref{propHS0a}. It is given by
\begin{equation}\label{MB1-J}
\begin{array}{l}
\phi_1(x) = \phi_1^{l} - \dfrac{(1-\lambda)(c_{10}^{l}- c_{10}^{a,l})H(x)}{z_2H(a)}+ \dfrac{(\phi_{0}^{l} - \phi_{0}^{a,l})}{(\ln c_{10}^{l}- \ln c_{10}^{a,l})} P_0(x)\\
\hspace*{.6in} - \dfrac{ \ln c_{10}^{l}- \ln c_{10}(x)}{ \big(  c_{10}^{l} -c_{10}^{a,l} \big)} N_{0},\\
c_{11}(x)= c_{11}^{l}+ \dfrac{\lambda z_1 - z_2}{z_2}\Big(c^2_{10}(x) - (c_{10}^{l})^2 \Big)-\dfrac{H(x)}{H(a)}M_{0},\\
J_{11}= \dfrac{M_{0}+z_1N_{0}}{H(a)}, \hspace*{.3in} J_{21}= -\dfrac{z_1}{z_2}\dfrac{M_{0}+z_2N_{0}}{H(a)}, 
\end{array}
\end{equation}
where $M_{0}$, $N_{0}$, and $P_0(x)$ are defined in \eqref{MB1-MNP}.
\end{lem}

\begin{proof} It follows from \eqref{1slow:28}, \eqref{1sT0}, and \eqref{1f0}  that
$$
\begin{aligned}
c_{11}(x)&= c_{11}^{l}+\dfrac{\lambda z_1 - z_2}{z_2}\big( c_{10}^2(x) - (c_{10}^{l})^2\big) + \dfrac{z_2T_1}{z_1 - z_2} H(x).
\end{aligned} 
$$
Then, from above and \eqref{MB1-MNP} at $x=a$ we have, 
\[ 
T_1 = -\dfrac{z_1 - z_2}{z_2 H(a)}M_{0}.
\]
Hence, $c_{11}(x)$ will be obtained. 
To find $\phi_1(x)$, one has from \eqref{1slow:28},
\[
\begin{aligned}
\phi_1(x) =& \phi_1^{l} + \dfrac{I_0}{z_1(z_1 - z_2)} \int_0^x \dfrac{c_{11}(s)}{h(s)c_{10}^2(s)}ds + \dfrac{(1-\lambda)T_0}{z_1 -z_2} H(x) - \dfrac{I_1}{z_1(z_1 - z_2)} \int_0^x \dfrac{1}{h(s)c_{10}(s)}ds.
\end{aligned}
\]
One can find the first integral from $c_{11}(x)$ and \eqref{1f0}.
Then from \eqref{1sI0} and above,
\[
\dfrac{I_0}{z_1(z_1-z_2)}\int_0^x \dfrac{c_{11}(s)}{h(s)c_{10}^2(s)}ds =\dfrac{(\phi_{0}^{l} - \phi_{0}^{a,l})}{(\ln c_{10}^{l}- \ln c_{10}^{a,l})} P_{0}(x).
\]

Hence,

\[
\begin{aligned}
\phi_1(x) =& \phi_1^{l} + \dfrac{\phi_{0}^{l} - \phi_{0}^{a,l}}{\ln c_{10}^{l}- \ln c_{10}^{a,l}} P_{0}(x) - \dfrac{(1-\lambda)(c_{10}^l- c_{10}^{a,l})}{z_2 H(a)} H(x) - \dfrac{I_1 H(a)}{z_1(z_1 - z_2)}\dfrac{\ln c_{10}^{l}- \ln c_{10}(x)}{c_{10}^l- c_{10}^{a,l}} .
\end{aligned}
\]

At $x=a$ one obtains,
\[
I_1= \dfrac{z_1(z_1 - z_2)}{ H(a)} N_0,
\]
where $N_0$ was defined in \eqref{MB1-MNP}. Then, the equations for $\phi_1(x), J_{11}$ and $J_{21}$ follow directly.
\end{proof}


\subsection{A singular orbit on $[a,b]$ where $Q(x)=Q$.}
Here we construct singular orbits for the connecting problem from $B_{1}$ to $B_2$. Each such an orbit will consist of two boundary layers $\Gamma^{a,m}$ at $x=a,~\Gamma^{b,m}$ at $x=b$, and a regular layer $\Lambda_2$ over the interval $[a,b]$.
\medskip

\noindent {\bf \underline{Fast dynamics and boundary/internal layers on $[a,b]$.} } 
By setting $\varepsilon =0$ in system \eqref{Sloweps1} with $Q=Q_2 \neq 0$, the slow manifold is
$$
\mathscr{Z}_2 = \Big\{ u=0,~ z_1c_1+ z_2c_2 + Q=0 \Big\}.
$$
The set of equilibria of \eqref{fzero02} is precisely $\mathscr{Z}_2$, and from lemma \eqref{0nh}, when $Q=Q_2$, the slow manifold $\mathscr{Z}_2$ is {normally hyperbolic} for system \eqref{feps02}.
We denote the stable \big(resp. unstable\big) manifold of $\mathscr{Z}_2$ by $W^s(\mathscr{Z}_2)$ \big(resp. $W^u(\mathscr{Z}_2)$\big). Let $M^{a,m}$ be the collection of orbits from $B_{1}$ in forward time under the flow of system \eqref{feps02} and $M^{b,m}$ be the collection of orbits from $B_2$ in backward time under the flow of system \eqref{feps02}. Then, for a singular orbit connecting $B_{1}$ to $B_2$, the boundary layer at $\tau =x = a$ must lie in $N^{a,m} = M^{a,m} \cap W^s(\mathscr{Z}_2)$ and the boundary layer at $\tau = x =b$ must lie in $N^{b,m} = M^{b,m} \cap W^u(\mathscr{Z}_2)$. 
As before we look for solutions
$$
\Gamma(\xi;d)= \big(\phi(\xi;d), u(\xi;d), c_1(\xi;d), c_2(\xi;d), J_1(d), J_2(d), \tau\big),
$$
of the system \eqref{fzero02} of the form \eqref{dsys}.
Substituting \eqref{dsys} into the system \eqref{fzero02}, we obtain, for the zeroth order in $d$,
\begin{equation}\label{0th12}
\begin{array}{l}
\phi_0'=u_0, \quad u_0'= -z_1c_{10}-z_2c_{20}- Q,\\
c_{10}'= -z_1c_{10}u_0,\quad c_{20}'= -z_2c_{20}u_0,\\
J_{10}'=J_{20}'=0, \quad \tau'=0,
\end{array}
\end{equation}
and from \eqref{fg2}, the first order in $d$ terms are,
\begin{equation}\label{1st12}
\begin{array}{l}
\phi_1'=u_1, \quad u_1'= -z_1c_{11}-z_2c_{21},\\
c_{11}'= -z_1c_{11}u_0 -z_1c_{10}u_1 +u_0 \Big((\lambda +1)z_2 c_{10}c_{20} + 2z_1 c_{10}^2 \Big),\\
c_{21}'= -z_2c_{21}u_0 -z_2c_{20}u_1 +u_0 \Big((\lambda +1)z_1 c_{10}c_{20} + 2\lambda z_2 c_{20}^2 \Big),\\
J_{11}'=J_{21}'=0, \quad \tau'=0.
\end{array}
\end{equation}


\begin{lem}{\label{le:0th12}}
The zeroth order system \eqref{0th12} has a complete set of first integrals, for $j=1,2,$
$$
\begin{array}{l}
\mathcal{H}_{10,m} = e^{z_1 \phi_0} c_{10}, \quad \mathcal{H}_{20,m} = e^{z_2 \phi_0} c_{20}, \quad \mathcal{H}_{30,m}=J_{10}, \quad \mathcal{H}_{40,m}=J_{20},\\
 \mathcal{H}_{50,m}=c_{10}+c_{20}-\dfrac{1}{2}u_0^2- Q \phi_0, \quad {\mathcal{H}}_{60,m}=\tau,
\end{array}
$$
and the first order system \eqref{1st12} has a complete set of first integrals,
$$
\begin{array}{l}
\mathcal{H}_{11,m} = z_1 \phi_1 + c_{11}/c_{10}+ 2c_{10} + (\lambda+1)c_{20},\quad
\mathcal{H}_{21,m} = z_2 \phi_1 +  c_{21}/c_{2  0}+ 2\lambda c_{20} + (\lambda+1)c_{10},\\
\mathcal{H}_{31,m}=u_0u_1 - c_{11} -c_{21} - (\lambda+1) c_{10}c_{20} - c_{10}^2 - \lambda c_{20}^2 + \phi_1 Q ,\quad
 \mathcal{H}_{41,m}=J_{11}, \quad \mathcal{H}_{51,m}=J_{21}.
\end{array}
$$
\end{lem}
\begin{proof}
It can be verified directly.
\end{proof}
\noindent Recall that we are interested in the solutions of $\Gamma^{a,m}(\xi;d) \subset N^{a,m} = M^{a,m} \cap W^s(\mathscr{Z}_2)$ with $\Gamma^{a,m}(0;d) \in B_{1}$, and $\Gamma^{b,m}(\xi;d) \subset N^{b,m} = M^{b,m} \cap W^u(\mathscr{Z}_2)$ with $\Gamma^{b,m}(0;d) \in B_{2}$.


\begin{prop}\label{propHSab}
Assume that $d>0$ is small.\\
(i) The stable manifold $W^s(\mathscr{Z}_2)$ intersects $B_1$ transversally at points
$$
\big(\phi_0^{a}+ \phi_1^{a}d + o(d), u_0^{a,m}+ u_1^{a,m}d+ o(d), c_{10}^{a}+ c_{11}^{a}d + o(d), c_{20}^{a}+c_{21}^{a}d + o(d), J_1(d), J_2(d), a\big),
$$
and the $\omega-$limit set of $N^{a,m}=M^{a,m}\cap W^s(\mathscr{Z}_2)$ is
$$
\omega\big(N^{a,m}\big)=\Big\{ \big( \phi_0^{a,m}+\phi_1^{a,m}d+o(d), 0, c_{10}^{a,m}+c_{11}^{a,m}d+o(d), c_{20}^{a,m}+c_{21}^{a,m}d+o(d), J_1(d), J_2(d), a\big)\Big\},
$$
where $J_k(d)=J_{k0}+J_{k1}d+o(d),~~k=1,2,$ can be arbitrary and $\phi_0= \phi_0^{a,m}$ is the unique solution of
$
z_1c_{10}^{a}e^{z_1(\phi_0^{a} -\phi_0)} + z_2c_{20}^{a}e^{z_2(\phi_0^{a} -\phi_0 )} + Q =0,
$
 and
 
\begin{equation}\label{prop2:47}
\begin{aligned}
&c_{10}^{a,m}=c_{10}^{a} e^{z_1\big(\phi_0^{a} -\phi_0^{a,m} \big)}, \quad c_{20}^{a,m}=c_{20}^{a} e^{z_2\big(\phi_0^{a} - \phi_0^{a,m} \big)}, \\
&u_0^{a,m}= sgn(z_1c_{10}^{a}+z_2c_{20}^{a})\sqrt{2\Big( \sum_{j=1}^2 \big( c_{j0}^{a}- c_{j0}^{a}e^{z_j(\phi_0^{a}- \phi_0^{a,m} )}\big)- Q(\phi_0^{a} - \phi_0^{a,m} )\Big)},\\
&\phi_1^{a,m} = \phi_1^{a} + \dfrac{z_1 c_{10}^{a,m}}{\sigma^{a,m}}\Bigg( \frac{c_{11}^{a}}{c_{10}^{a}} - \frac{c_{21}^{a}}{c_{20}^{a}} +(1-\lambda )\Big(c_{10}^{a}+c_{20}^{a} -c_{10}^{a,m} -c_{20}^{a,m} \Big)\Bigg)\\
&\hspace*{.5in}- \dfrac{1}{\sigma^{a,m}}\Big(\frac{c_{21}^{a}}{c_{20}^{a}}+ 2\lambda c_{20}^{a}+ (\lambda +1) c_{10}^{a} - 2\lambda c_{20}^{a,m} - (\lambda+ 1) c_{10}^{a,m} \Big)Q,\\
&z_1c_{11}^{a,m}=-z_2c_{21}^{a,m}= \dfrac{z_1c_{10}^{a,m}}{\sigma^{a,m}}\big(z_1c_{10}^{a,m}+ Q\big) \Bigg(z_1 \dfrac{c_{21}^{a}}{c_{20}^{a}} -z_2 \dfrac{c_{11}^{a}}{c_{10}^{a}} +(z_1 -z_2)w(c_{10}^{a} , c_{20}^{a}) \\
&\hspace*{1.3in}+\dfrac{1}{z_2}\Big(\big[2(\lambda z_1- z_2) + (1-\lambda) z_2\big]Q + 2(\lambda z_1- z_2)(z_1-z_2)c_{10}^{a,m}\Big)\Bigg),\\
&u_1^{a,m}= \dfrac{1 }{u_0^{a,m}}\Big(( c_{10}^{a}+ c_{20}^{a})(c_{10}^{a}+ \lambda  c_{20}^{a}) - ( c_{10}^{a,m}+ c_{20}^{a,m})(c_{10}^{a,m}+ \lambda  c_{20}^{a,m})\\
& \qquad \qquad \qquad +  c_{11}^{a}+c_{21}^{a} - c_{11}^{a,m}- c_{21}^{a,m} + \big(\phi_1^{a,m}- \phi_1^{a}\big) Q \Big),
\end{aligned}
\end{equation}
where $\sigma$ is defined in \eqref{Abbrev}.\\
(ii) The unstable manifold $W^u(\mathscr{Z}_2)$ intersects $B_2$ transversally at points
$$
\big(\phi_0^{b}+\phi_1^{b}d + o(d), u_0^{b,m}+ u_1^{b,m}d+ o(d), c_{10}^{b}+ c_{11}^{b}d + o(d), c_{20}^{b}+c_{21}^{b}d + o(d), J_1(d), J_2(d), b\big),
$$
and the $\alpha-$limit set of $N^{b,m}=M^{b,m}\cap W^u(\mathscr{Z}_2)$ is
$$
\alpha\big(N^{b,m}\big)=\Big\{ \big( \phi_0^{b,m}+\phi_1^{b,m}d+o(d), 0, c_{10}^{b,m}+c_{11}^{b,m}d+o(d), c_{20}^{b,m}+c_{21}^{b,m}d+o(d), J_1(d), J_2(d), b\big)\Big\},
$$
where $J_k(d)=J_{k0}+J_{k1}d+o(d), k=1,2,$ can be arbitrary and $\phi_0= \phi_0^{b,m}$ is the unique solution of
$$
z_1c_{10}^{b}e^{z_1( \phi_0^{b} - \phi_0)} + z_2c_{20}^{b}e^{z_2( \phi_0^{b} - \phi_0)} + Q =0.
$$
Besides,
\begin{equation}\label{prop2:49}
\begin{array}{l}
c_{10}^{b,m}=c_{10}^{b} e^{z_1\big(\phi_0^{b} - \phi_0^{b,m} \big)}, \quad c_{20}^{b,m}=c_{20}^{b} e^{-z_2\big(\phi_0^{b,m} -\phi_0^{b} \big)}, \\
u_0^{b,m}= -sgn(z_1c_{10}^{b}+z_2c_{20}^{b})\sqrt{2\Big(\displaystyle \sum_{j=1}^2 \big( c_{j0}^{b}- c_{j0}^{b}e^{z_j(\phi_0^{b} - \phi_0^{b,m} )}\big)- Q(\phi_0^{b} - \phi_0^{b,m} )\Big)},\\
\phi_1^{b,m} = \phi_1^{b} + \dfrac{z_1 c_{10}^{b,m}}{\sigma^{b,m}}\Bigg( \frac{c_{11}^{b}}{c_{10}^{b}} - \frac{c_{21}^{b}}{c_{20}^{b}} +(1-\lambda )\Big(c_{10}^{b}+c_{20}^{b} -c_{10}^{b,m} -c_{20}^{b,m} \Big)\Bigg)\\
\hspace*{.5in}- \dfrac{1}{\sigma^{b,m}}\Big(\frac{c_{21}^{b}}{c_{20}^{b}}+ 2\lambda c_{20}^{b}+ (\lambda +1) c_{10}^{b} - 2\lambda c_{20}^{b,m} - (\lambda+ 1) c_{10}^{b,m} \Big)Q,\\
z_1c_{11}^{b,m}=-z_2c_{21}^{b,m} =\dfrac{ z_1c_{10}^{b,m}}{\sigma^{b,m}}\big(z_1c_{10}^{b,m}+ Q\big) \Bigg(z_1 \dfrac{c_{21}^{b}}{c_{20}^{b}} -z_2 \dfrac{c_{11}^{b}}{c_{10}^{b}}  +(z_1 -z_2)w(c_{10}^{b} , c_{20}^{b})\\
\hspace*{1.3in} +\dfrac{1}{z_2}\Big(\big[2(\lambda z_1- z_2) + (1-\lambda) z_2\big]Q + 2(\lambda z_1- z_2)(z_1-z_2)c_{10}^{b,m} \Big) \Bigg),\\
u_1^{b,m}= \dfrac{1}{u_0^{b,m}}\Big(( c_{10}^{b}+ c_{20}^{b})(c_{10}^{b}+ \lambda  c_{20}^{b}) - ( c_{10}^{b,m}+ c_{20}^{b,m})(c_{10}^{b,m}+ \lambda  c_{20}^{b,m})\\
 \qquad \qquad \qquad 
 + \big( c_{11}^{b}+c_{21}^{b} - c_{11}^{b,m}- c_{21}^{b,m}\big)+ \big(\phi_1^{b,m}- \phi_1^{b}\big) Q\Big).
\end{array}
\end{equation}

\begin{proof}
It is almost identical to the proof of Prop 2.2.
\end{proof} 

\end{prop}

\noindent {\bf \underline{Limiting slow dynamics and regular layers on $[a,b]$.}}
Next we construct the regular layer on $\mathscr{Z}_2$ that connects $\omega(N^{a,m})$ and $\alpha(N^{b,m})$. Note that, for $\varepsilon =0$, system \eqref{Sloweps1} with $Q=Q_2 \neq 0$ loses most information. As before, we make a rescaling $u=\varepsilon p$ and $-z_2c_2=z_1c_1+ Q+ \varepsilon q$ in system \eqref{Sloweps1}. In terms of the new variables, system \eqref{Sloweps1} with $Q\neq 0$ becomes
\begin{equation}\label{2slow:1}
\begin{array}{l}
\overset{.}{\phi} = p, \quad \varepsilon\overset{.}{p} = q- \varepsilon \dfrac{h_{\tau}(\tau)}{h(\tau)}p, \quad \varepsilon \overset{.}{q}= (z_1f_1 +z_2f_2) p + \dfrac{z_1g_1+z_2g_2}{h(\tau)},  \\
 \overset{.}{c}_1 = -f_1p - \dfrac{g_1}{h(\tau)}, \quad \overset{.}{J_1} = \overset{.}{J_2} = 0, \quad \overset{.}{\tau}=1,
\end{array}
\end{equation}
where, for $k=1,2,$ $
f_k=f_k\big( c_1, -\frac{z_1c_1+Q+\varepsilon q}{z_2};d, \lambda d \big) $, and$ g_k=g_k\big( c_1, -\frac{z_1c_1+Q+\varepsilon q}{z_2}, J_1, J_2;d, \lambda d \big),
$
were defined in \eqref{fg2}.
Note that the equations in \eqref{2slow:1} are consistent with those in the paper \cite{EL07}. 
The system \eqref{2slow:1} is also a singular perturbation problem and its limiting slow system is
\begin{equation}\label{51}
\begin{array}{l}
q = 0, \quad p = - \dfrac{\sum_{j=1}^2 z_jg_j\big(c_1,-\frac{z_1c_1+Q}{z_2}, J_1, J_2; d, \lambda d \big)}{h(\tau) \Big((z_1-z_2)z_1c_1 -z_2Q+2\big((1-\lambda)z_1c_1-\lambda Q\big)Qd \Big) +o(d)}, \\
\overset{.}{\phi}=p,\quad 
\overset{.}c_1= -f_1\big(c_1,-\frac{z_1c_1+Q}{z_2}; d, \lambda d \big)p - \dfrac{1}{h(\tau)} g_1\big(c_1,-\frac{z_1c_1+Q}{z_2}, J_1, J_2; d, \lambda d \big),\\
\overset{.}{J_1} = \overset{.}{J_2} = 0, \quad \overset{.}{\tau}=1.
\end{array}
\end{equation}
In the above, we have applied \eqref{fg2} (with $\varepsilon =0$) for the expression in $p$ to obtain,
$$
z_1f_1\big(c_1, -\frac{z_1c_1+Q}{z_2}, d \big) + z_2 f_2\big(c_1, -\frac{z_1c_1+Q}{z_2}, d \big)= (z_1-z_2)z_1c_1 -z_2Q+2\big((1-\lambda)z_1c_1-\lambda Q\big)Qd +o(d).
$$
From system \eqref{51}, the slow manifold is
$$\displaystyle
\mathcal{S}= \Big\{ q=0,~p = - \dfrac{\sum_{j=1}^2 z_jg_j\big(c_1,-\frac{z_1c_1+Q}{z_2}, J_1, J_2; d, \lambda d \big)}{h(\tau) \Big((z_1-z_2)z_1c_1 -z_2Q+2\big((1-\lambda)z_1c_1-\lambda Q\big)Qd \Big) +o(d)} \Big\}.
$$
Therefore, the limiting slow system on $\mathcal{S}$ is
\begin{equation}\label{52}
\begin{array}{l}
\overset{.}{\phi}=p,\quad
\overset{.}c_1= -f_1\big(c_1,-\frac{z_1c_1+Q}{z_2}; d, \lambda d \big)p - \dfrac{1}{h(\tau)} g_1\big(c_1,-\frac{z_1c_1+Q}{z_2}, J_1, J_2; d, \lambda d \big),\\
\overset{.}{J_1} = \overset{.}{J_2} = 0, \quad \overset{.}{\tau}=1.
\end{array}
\end{equation}
where $p$ is defined in \eqref{51}.
As for the layer problem, we look for solutions of \eqref{52} of the form \eqref{1slow:dsys} to connect $\omega\big(N^{a,m} \big)$ and $\alpha\big(N^{b,m} \big)$ given in Proposition \eqref{propHSab}. In particular, for $j=0,1,$ we have,
$
( \phi_j(a),c_{1j}(a)) = (\phi_j^{a,m}, c_{1j}^{a,m})$, and $( \phi_j(b),c_{1j}(b)) = (\phi_j^{b,m}, c_{1j}^{b,m}).
$
Note, from  \eqref{fg2}  over the slow manifold $\mathscr{Z}_2$, that
$$
\begin{aligned}
&f_1\big(c_1,-\frac{z_1c_1+Q}{z_2}; d, \lambda d \big)=  z_1c_{10} + \Big(z_1 c_{11} + (\lambda - 1) z_1c_{10}^2 + (1+\lambda) c_{10}Q\Big) d + o(d) ,\\
&g_1\big(c_1,-\frac{z_1c_1+Q}{z_2}; d, \lambda d \big)=  J_{10} + \Big(J_{11} - c_{10}( J_{10} + J_{20} ) - c_{10} ( J_{10} + \lambda J_{20} ) \Big) d + o(d) ,\\
&z_1g_1+z_2g_2 = I_0 + \Big(I_1
 + \big((\lambda - 1)z_1c_{10} + \lambda Q \big)T_0 +\Lambda_0 Q \Big  ) d +o(d).
\end{aligned}
$$
It follows from \eqref{52}, along with above terms, that the zeroth and first order terms are, 
\begin{equation}\label{Mzeroth}
\begin{array}{l}
\overset{.}{\phi}_0= - \dfrac{I_0}{h(\tau)\sigma(\tau)}, \quad \overset{.}{c}_{10}=  \dfrac{z_1z_2c_{10}T_0+ z_2J_{10}Q}{h(\tau)\sigma(\tau)},\\
\overset{.}{J}_{10} = \overset{.}{J}_{20} = 0, \quad \overset{.}{\tau} =1,
\end{array}
\end{equation}
and
\begin{equation}\label{Mfirst}
\begin{array}{l}
\overset{.}{\phi}_1=  \dfrac{I_0}{ h(\tau) \sigma^2(\tau)}\Big( (z_1 - z_2) z_1 c_{11} + 2(1-\lambda)z_1c_{10}Q - 2\lambda Q^2\Big)\\
\hspace{.4in} + 
\dfrac{1}{h(\tau) \sigma(\tau)}\Big(\big((1-\lambda)z_1c_{10} - \lambda Q\big)T_0 - I_1- \Lambda_0 Q\Big),\vspace*{.075in} \\ 
\overset{.}{c}_{11}= \dfrac{2(\lambda z_1 - z_2) z_1c_{10}^2T_0+ z_1z_2c_{10}T_1 }{h(\tau) \sigma(\tau)} + \dfrac{  2\lambda z_1c_{10} T_0 + 2(z_1-z_2)c_{10} J_{10} + z_2J_{11} }{h(\tau) \sigma(\tau)}Q\\
\hspace{.4in} - \dfrac{I_0Q }{h(\tau)\sigma^2(\tau)}\Big( z_1 z_2c_{11} + z_1c_{10}\big( 2(1-\lambda)z_1c_{10} - 2\lambda Q \big) \Big),\vspace*{.075in} \\ 
\overset{.}{J}_{11} = \overset{.}{J}_{21} = 0.
\end{array}
\end{equation}

\begin{rem}
Note that to find $\dot{c}_1$ in \eqref{Mfirst}, from \eqref{52}  we have,
$$
\begin{aligned}
\dot{c}_{11}(\tau)=&\dfrac{z_1c_{10}}{h(\tau) \sigma(\tau) } \Big( {I}_1 - \big((1-\lambda)z_1c_{10} - \lambda Q\big){T}_0 + \Lambda_0 Q\Big)\\
&+\dfrac{{I}_0}{h(\tau)\sigma(\tau)}\Big(z_1 c_{11} - (1-\lambda) z_1c_{10}^2 + (1+\lambda) c_{10}Q\Big)\\
&-\dfrac{1}{h(\tau) \sigma(\tau)}\Big(z_1(z_1 -z_2)c_{10} - z_2 Q\Big)\Big(  J_{11} -c_{10}({T}_0 + \Lambda_0 ) \Big)\\
&-\dfrac{z_1c_{10}{I}_0}{h(\tau) \sigma^2(\tau)}\Big(z_1 c_{11}(z_1 - z_2)+ 2(1-\lambda)z_1c_{10}Q - 2\lambda Q^2\Big).
\end{aligned}
$$
\qed
\end{rem}

\begin{lem}\label{LemM0th}
There is a unique solution $(\phi_0(y), c_{10}(y), J_{10}, J_{20}, \tau(y)\big)$ of \eqref{Mzeroth} such that 
$$
\big( \phi_0(0), c_{10}(0), \tau(0) \big) = \big( \phi_0^{a,m}, c_{10}^{a,m}, a \big) ~~\textrm{and}~~\big( \phi_0(y^*), c_{10}(y^*), \tau(y^*) \big) = \big( \phi_0^{b,m}, c_{10}^{b,m}, b \big) 
$$
for some $y^* >0$, where $\phi_0^{a,m},~\phi_0^{b,m},~c_{10}^{a,m}$, and $c_{10}^{b,m}$ are given in Proposition \eqref{propHSab}. This solution is given by 
$$
\begin{array}{l}
\phi_0(y)= \phi_0^{a,m} - I_0y, \quad
c_{10}(y)= e^{z_1z_2T_0y}c_{10}^{a,m}+\dfrac{J_{10}Q}{z_1T_0}\Big( e^{z_1z_2T_0y}-1 \Big), \vspace*{.075in} \\ 
T_0 = - \dfrac{(z_1 - z_2) \big(c_{10}^{a,m} - c_{10}^{b,m} \big) +  z_2\big(\phi_0^{a,m} - \phi_{0}^{b,m} \big)Q}{z_2\big(H(b)- H(a)\big)}.
\end{array}
$$
\begin{proof} The complete proof is on \cite{EL07}. Here, we express the proof briefly:
Since $h(\tau)> 0$ and $-z_2c_{20}=z_1c_{10}+Q > 0$, then the system \eqref{Mzeroth} has the same phase portrait as that of the following system obtained by multiplying $h(\tau)\sigma (\tau)$ on the right-hand side of the system \eqref{Mzeroth},
\begin{equation}\label{2sB:1}
\begin{array}{l}
\dfrac{d}{dy}\phi_0= - I_0, \quad 
\dfrac{d}{dy}c_{10}= z_1z_2T_0c_{10}+ z_2QJ_{10}, \vspace*{.075in} \\ 
\dfrac{d}{dy}J_{10} = \dfrac{d}{dy} J_{20}=0, \quad \dfrac{d}{dy} \tau = h(\tau)\sigma.
\end{array}
\end{equation}
The expressions for $\phi_0(y)$ and $c_{10}(y)$ directly obtain from above with the initial values $(\phi_0^{a,m}, c_{10}^{a,m},J_{k0}, a)$. It also follows from last equation that,
$$
\begin{array}{l}\displaystyle
\int_{a}^{\tau} \dfrac{1}{h(s)}ds =\dfrac{(z_1 -z_2)c_{10}^{a,m}}{z_2T_0}\Big(e^{z_1z_2T_0y}-1 \Big)+\dfrac{(z_1 -z_2)J_{10}Q}{T_0}\Big(\dfrac{ e^{z_1z_2T_0y}-1}{z_1z_2T_0}-y \Big)- z_2Qy.
\end{array}
$$
Assume $\tau(y^*) = b$ for some $y^* >0$, then, $\phi_0(y^*) = \phi_0^{b,m}$ and $c_{10}(y^*) = c_{10}^{b,m}$, and the proof is complete.
\end{proof}
\end{lem}


\noindent To find the first order terms in $d$, i.e., ion size, and for convenience we first introduce some integrals integrals. First note that from \eqref{Abbrev} and \eqref{2sB:1} one has,

\begin{equation}\label{Int-Z}
\begin{aligned}
\int_0^y Z_0(s)ds =& z_1z_2{T}_0 \int_0^y \dfrac{{I}_0 Q}{(z_1 -z_2)\big(z_1c_{10}^{a,m} {T}_0 + J_{10}Q \big)e^{z_1z_2{T}_0s} -{I}_0 Q}ds\\
 = & \ln \Big(\Big[z_1(z_1 -z_2)c_{10}^{a,m} {T}_0 + (z_1-z_2)J_{10}Q \Big]e^{z_1z_2{T}_0 y} - {I}_0 Q \Big)\\
&\quad - \ln \big(z_1(z_1 -z_2)c_{10}^{a,m} {T}_0 -z_2 {T}_0 Q\big) - z_1z_2 {T}_0 y \\
=& \ln \sigma(y) - \ln \sigma^{a,m} - z_1z_2{T}_0 y .
\end{aligned}
\end{equation}

Then, for $S_z:= \int_0^y \dfrac{z_1{I}_0 c_{10}(s)}{\sigma(s)}ds$, it follows from above that,
\begin{equation}\label{Sz}
\begin{aligned}
S_z =& \dfrac{{I}_0}{(z_1 -z_2)}\int_0^y \big(1+ \dfrac{z_2Q}{\sigma(s)} \big) ds= J_{10} y +\dfrac{\ln \sigma(y) - \ln \sigma^{a,m}}{z_1(z_1-z_2)}.
\end{aligned}
\end{equation}

Moreover, we introduce $S_j = S_j(y,Q)$, for $j=1,\cdots,5$, as follows,
\begin{equation}\label{S-integ}
\begin{aligned}
&S_1=\int_0^y c_{10}(s)ds = \dfrac{1}{z_1z_2{T}_0 } \Big( c_{10}(y) - c_{10}^{a,m} -z_2J_{10}Qy \Big),\\
&S_2=\int_0^y c_{10}^2(s) ds = \dfrac{1}{2z_1z_2{T}_0} \Big(c_{10}^2(y) - \big(c_{10}^{a,m}\big)^2\Big) -\dfrac{J_{10}}{z_1{T}_0}S_1Q,\\
&S_3= \int_0^y c_{10}(s)e^{-z_1z_2{T}_0 s}ds = c_{10}^{a,m}y+ \dfrac{J_{10}Q}{z_1{T}_0}\Big(y+ \dfrac{e^{-z_1z_2{T}_0 y} - 1}{z_1z_2{T}_0 }\Big),\\
& S_4= \int_0^y c_{10}^2(s)e^{-z_1z_2 {T}_0 s}ds = \dfrac{1}{z_1z_2{T}_0}\Big(c_{10}^2(y)e^{-z_1z_2{T}_0 y} -(c_{10}^{a,m})^2 -2z_2J_{10}QS_3 \Big),\\
& S_5=\int_0^y \dfrac{z_1{I}_0 c_{10}^2(s)}{\sigma(s)}ds = \dfrac{{I}_0 \big(c_{10}(y) - c_{10}^{a,m}\big)}{z_1z_2(z_1 -z_2){T}_0} +  \dfrac{z_2\big( \ln \sigma(y) - \ln  \sigma^{a,m} \big)}{z_1^2(z_1 -z_2)^2}Q - \dfrac{J_{10}^2yQ}{z_1{T}_0}.
\end{aligned}
\end{equation}

\begin{rem}
To obtain the expressions for $S_1$ to $S_3$ in \eqref{S-integ} in above, one should simply use \eqref{2sB:1}.
To find $S_4$, one may apply integration by parts to obtain the following,
$$
\begin{aligned}
S_4 =& \dfrac{(c_{10}^{a,m})^2 - c_{10}^2(y)e^{-z_1z_2{T}_0y} }{z_1z_2{T}_0} + \dfrac{2}{z_1z_2{T}_0} \int_0^y c_{10}(s)\Big(z_1z_2{T}_0c_{10}(s) + z_2J_{10}Q\Big)e^{-z_1z_2{T}_0 s}ds\\
& = \dfrac{(c_{10}^{a,m})^2 - c_{10}^2(y)e^{-z_1z_2{T}_0 y} }{z_1z_2{T}_0} + 2 \int_0^y c_{10}^2(s)e^{-z_1z_2{T}_0 s}ds + \dfrac{2J_{10}Q}{z_1{T}_0} S_3.
\end{aligned}
$$
To find $S_5$, note that
$$
\begin{aligned}
S_5 =&  \dfrac{1}{z_2{T}_0}\int_0^y\dfrac{{I}_0 c_{10}(s)\big(\dot{c}_{10} - z_2J_{10}Q\big)}{\sigma(s)}ds= \dfrac{{I}_0}{z_2{T}_0} \int_0^y \dfrac{c_{10}(s)\dot{c}_{10}(s) }{\sigma(s)}ds  -\dfrac{J_{10}Q}{z_1{T}_0} \int_0^y \dfrac{z_1{I}_0c_{10}(s)}{\sigma(s)}ds.
\end{aligned}
$$
Substitute the second integral by  $S_z(y)$ in \eqref{Sz} and simplify to obtain $S_5$.
\qed
\end{rem}


\medskip

\begin{lem}{\label{lem:2slow1st}} 
There is a unique solution $\big(\phi_1(y), c_{11}(y), J_{11}, J_{21}, \tau(y) \big)$ of \eqref{Mfirst} such that
$$
\big(\phi_1(0), c_{11}(0), \tau(0) \big) = (\phi_1^{a,m}, c_{11}^{a,m}, a)~~ \textrm{~and~}~~\big(\phi_1(y^*), c_{11}(y^*), \tau(y^*) \big) = (\phi_1^{b,m}, c_{11}^{b,m}, b),
$$
for the same $y^* >0$ in previous lemma, and where $\phi_1^{a,m}, \phi_1^{b,m}, c_{11}^{a,m}$ and $c_{11}^{b,m}$ are given in Proposition \eqref{propHSab}. The solution is given by

\[
\begin{aligned}
\phi_1(y)=& \phi_1^{a,m}+\dfrac{\lambda-1}{z_2^2 {T}_0} \Big((z_1-z_2)J_{10}+{I}_0\Big)(c_{10}(y) - c_{10}^{a,m})  - {I}_1 y \\
& + \dfrac{(z_1-z_2)I_0}{z_2T_0 \sigma(y)}\Bigg \{ \dfrac{z_1T_0\big(c_{10}(y) - c_{10}^{a,m} \big) c_{11}^{a,m}}{\big(z_1{T}_0 c_{10}^{a,m}+J_{10}Q \big)} - \dfrac{(\lambda z_1 - z_2)}{z_2}\Big(c_{10}^2(y) - \big(c_{10}^{a,m}\big)^2\Big)\\
 &\qquad \qquad \qquad - z_1z_2{T}_1 \Big( S_1-e^{z_1z_2{T}_0 y} S_3\Big) +2z_1(z_1-z_2) T_0 e^{z_1z_2{T}_0 y} S_4\\
& \qquad \qquad \qquad - 2z_1 \Lambda_0  Q  S_1 -z_2 J_{11} yQ + \dfrac{J_{11}  \big(c_{10}(y)- c_{10}^{a,m}\big)}{\big(z_1{T}_0 c_{10}^{a,m} +J_{10}Q\big)} Q\\
&    \qquad \qquad \qquad +2(z_1-z_2) J_{10} e^{z_1z_2 {T}_0 y}Q S_3 \Bigg \} +\dfrac{2(z_1-z_2)\Lambda_0}{z_2{T}_0}J_{10} yQ,
\end{aligned}
\]
\[
\begin{aligned}
c_{11}(y)=&~ c_{11}^{a,m}+z_1z_2{T}_1 S_1+2z_1(\lambda z_1 - z_2){T}_0 S_2 +z_2J_{11}Qy \\
& + \Big( 2\lambda z_1{T}_0 + 2(z_1-z_2)J_{10}\Big)S_1 Q \\
& + \dfrac{I_0}{T_0 \sigma(y)}\Bigg \{ \dfrac{z_1T_0\big(c_{10}^{a,m}- c_{10}(y) \big) c_{11}^{a,m}Q}{\big(z_1{T}_0 c_{10}^{a,m}+J_{10}Q \big)} +2z_1(\lambda z_1 - z_2)T_0 QS_2\\
&\qquad \qquad \quad +z_1z_2{T}_1 Q\Big( S_1-e^{z_1z_2{T}_0 y} S_3\Big) - 2z_1(z_1-z_2)T_0 e^{z_1z_2{T}_0 y}Q S_4\\
& \qquad \qquad \quad + \Big( 2\lambda z_1{T}_0 + 2(z_1-z_2)J_{10}\Big) Q^2  S_1  +z_2J_{11} yQ^2\\
&\qquad \qquad \quad + \dfrac{J_{11} }{z_1{T}_0} \Big(1 -e^{z_1z_2{T}_0 y} \Big)Q^2 - 2(z_1-z_2)J_{10} e^{z_1z_2{T}_0 y}Q^2S_3\Bigg \},
\end{aligned}
\]
and,
\begin{equation}\label{70J_1}
\begin{aligned}
\dfrac{ T_1 }{T_0}(c_{10}^{b,m} - c_{10}^{a,m})-& 
\dfrac{T_1 }{(z_1-z_2) T_0} z_2y^* I_0Q
 =  c_{11}^{b,m}- c_{11}^{a,m} - 2(\lambda z_1 - z_2) z_1 {T}_0 S_2(y^*)  \\
& -\Big(\big(2\lambda +\dfrac{(1-\lambda) z_2}{z_1-z_2} \big) z_1{T}_0 + 2(z_1-z_2) J_{10}\Big)S_1(y^*)Q\\
&  + \dfrac{z_2(\phi_1^{b,m} - \phi_1^{a,m})Q}{z_1-z_2} + 2(1-\lambda)z_1S_5(y^*)Q   -2\lambda S_z(y^*)Q^2 \\
& - \dfrac{2(1-\lambda)z_2}{z_1-z_2} S_z(y^*)Q^2  - \dfrac{2\lambda z_2 {T}_0 y^*}{(z_1-z_2)} Q^2 + \dfrac{z_2}{z_1-z_2}\Lambda_0 y^*Q^2.
\end{aligned}
\end{equation} 
where $\sigma(y), {T}_0 $, ${T}_1$, ${I}_0$ and ${I}_1$ were defined in \eqref{Abbrev}. Moreover, $S_z$ and $S_{k}(y)$'s, for $k=1,\cdots,5$,  were defined in \eqref{Sz} and \eqref{S-integ} respectively.
\end{lem}

\begin{proof}
We postpone the proof until Appendix Section \ref{sec-Appendix}.
\end{proof}


\subsection{A singular orbit on $[b,1]$ where $Q(x)=0$.}
The construction of singular orbits on the interval $[b,1]$ is virtually identical to the structure of singular orbits on $[0,a]$ in section \eqref{sec-sing-0a}. We merely state the results for later use.
\medskip

\noindent { \bf \underline{Fast Dynamics and Boundary/Internal Layers on $[b,1]$.} } 
The slow manifold is,
$$
\mathscr{Z}_3 = \Big\{ u=0,~ z_1c_1+ z_2c_2 =0 \Big\}.
$$
The limiting fast system is,
\[
\begin{array}{l}
\phi' = u, \quad u' = -z_1c_1 -z_2c_2, \quad
c_1' = -f_1(c_1, c_2; d_1, d_2)u,\\
c_2' = -f_2(c_1, c_2; d_1, d_2)u, \quad
J_1' = J_2' = 0, \quad \tau'=0.
\end{array}
\]
Recall that we are interested in the solutions of $\Gamma^{b,r}(\xi;d) \subset N^{b,r} = M^{b,r} \cap W^s(\mathscr{Z}_3)$ with $\Gamma^{b,r}(0;d) \in B_{2}$, and $\Gamma^{r}(\xi;d) \subset N^{r} = M^{r} \cap W^u(\mathscr{Z}_3)$ with $\Gamma^{r}(0;d) \in B_{3}$.


\begin{prop}{\label{propHSb1}}
Assume that $d>0$ is small.\\
(i) The stable manifold $W^s(\mathscr{Z}_3)$ intersects $B_2$ transversally at points
$$
\big(\phi^{b}, u_0^{b,r}+ u_1^{b,r}d+ o(d), c_{10}^{b}+ c_{11}^{b}d, c_{20}^{b}+ c_{21}^{b}d, J_1(d), J_2(d), b\big),
$$
and the $\omega-$limit set of $N^{b,r}=M^{b,r}\cap W^s(\mathscr{Z}_3)$ is
$$
\omega\big(N^{b,r}\big)=\Big\{ \big( \phi_0^{b,r}+\phi_1^{b,r}d+o(d), 0, c_{10}^{b,r}+c_{11}^{b,r}d+o(d), c_{20}^{b,r}+c_{21}^{b,r}d+o(d), J_1(d), J_2(d), b\big)\Big\},
$$
where $J_k(d)=J_{k0}+J_{k1}d+o(d),~~k=1,2,$ can be arbitrary and
\begin{equation}\label{prop3:61}
\begin{array}{l}
\phi_0^{b,r} = \phi_0^{b} - \frac{1}{z_1 -z_2}\ln \dfrac{-z_2 c_{20}^{b}}{z_1c_{10}^{b}}, \hspace*{.3in} z_1c_{10}^{b,r}=-z_2c_{20}^{b,r}=\big(z_1c_{10}^{b} \big)^{\frac{-z_2}{z_1-z_2}}\big(-z_2c_{20}^{b} \big)^{\frac{z_1}{z_1-z_2}},\\
u_0^{b,r}= -sgn(z_1c_{10}^{b}+z_2c_{20}^{b})\sqrt{2\Big(c_{10}^{b}+c_{20}^{b}+ \frac{z_1 -z_2}{z_1z_2} \big(z_1c_{10}^{b} \big)^{\frac{-z_2}{z_1-z_2}}\big(-z_2c_{20}^{b} \big)^{\frac{z_1}{z_1-z_2}} \Big)},\\
\phi_1^{b,r} = \phi_1^{b} + \dfrac{1}{z_1-z_2 }\Big(  \frac{c_{11}^{b}}{c_{10}^{b}} - \frac{c_{21}^{b}}{c_{20}^{b}} +(1-\lambda )\big(c_{10}^{b}+c_{20}^{b} -c_{10}^{b,r} -c_{20}^{b,r} \big) \Big),\\
z_1c_{11}^{b,r}=-z_2c_{21}^{b,r}=z_1c_{10}^{b,r}\Big(\frac{1}{z_1-z_2 }\Big( z_1 \frac{c_{21}^{b}}{c_{20}^{b}} -z_2 \frac{c_{11}^{b}}{c_{10}^{b}} \Big) + w(c_{10}^{b}, c_{20}^{b}) + \frac{2(\lambda z_1 -z_2)}{z_2}c_{10}^{b,r} \Big),\\
u_1^{b,r}= \frac{1}{u_0^{b,r}}\Big( ( c_{10}^{b}+ c_{20}^{b})(c_{10}^{b}+ \lambda  c_{20}^{b}) - ( c_{10}^{b,r}+ c_{20}^{b,r})(c_{10}^{b,r}+ \lambda  c_{20}^{b,r})\\
\qquad \qquad \qquad + \big( c_{11}^{b}+c_{21}^{b} - c_{11}^{b,r}- c_{21}^{b,r}\big) \Big).
\end{array}
\end{equation}
where, $w$ is defined in \eqref{Abbrev}.\\
(ii) The unstable manifold $W^u(\mathscr{Z}_3)$ intersects $B_3$ transversally at points
$$
\big(0, u_0^{r}+ u_1^{r}d+ o(d), r_1, r_2, J_1(d), J_2(d), 1\big),
$$
and the $\alpha-$limit set of $N^{r}=M^{r}\cap W^u(\mathscr{Z}_3)$ is
$$
\alpha\big(N^{r}\big)=\Big\{ \big( \phi_0^{r}+\phi_1^{r}d+o(d), 0, c_{10}^{r}+c_{11}^{r}d+o(d), c_{20}^{r}+c_{21}^{r}d+o(d), J_1(d), J_2(d), 1\big)\Big\},
$$
where $J_k(d)=J_{k0}+J_{k1}d+o(d),~~k=1,2,$ can be arbitrary and
\[
\begin{array}{l}
\phi_0^{r} =  - \frac{1}{z_1 -z_2}\ln \frac{-z_2 r_2}{z_1r_1}, \quad z_1c_{10}^{r}=-z_2c_{20}^{r}=\big(z_1r_1 \big)^{\frac{-z_2}{z_1-z_2}}\big(-z_2r_2 \big)^{\frac{z_1}{z_1-z_2}},\\
u_0^{r}= -sgn(z_1r_1+z_2r_2)\sqrt{2\Big(r_1+r_2+ \frac{z_1 -z_2}{z_1z_2} \big(z_1r_1 \big)^{\frac{-z_2}{z_1-z_2}}\big(-z_2r_2 \big)^{\frac{z_1}{z_1-z_2}} \Big)},\\
\phi_1^{r}= \frac{1-\lambda}{z_1 -z_2} \big(r_1+r_2-c_{10}^{r}-c_{20}^{r} \big),\\
z_1c_{11}^{r}=-z_2c_{21}^{r}=z_1c_{10}^{r}\Big(w(r_1,r_2) + \frac{2(\lambda z_1 -z_2)}{z_2}c_{10}^{r} \Big),\\
u_1^{r} =\frac{1}{u_0^{r}}\Big( (r_1+r_2)(r_1+\lambda r_2) - (c_{10}^{r}+c_{20}^{r})(c_{10}^{r}+\lambda c_{20}^{r}) - c_{11}^{r} -c_{21}^{r} \Big).
\end{array}
\]
\end{prop}

\begin{proof}
It is almost identical to the proof of Prop 2.2.
\end{proof}

\noindent {\bf \underline{Limiting slow dynamics and regular layers on $[b, 1]$.}}
We now examine the existence of regular layers or outer solutions that connects $\omega(N^{b,r})$ to $\alpha(N^{r})$. Following exactly the same analysis for the limiting slow dynamics in section \eqref{sec-sing-0a}, one obtains the following Lemma.

\begin{lem}\label{lem:3zeroth}
For $b \leq x \leq 1$, there is a unique solution $(\phi_0(x), c_{10}(x), J_{10}, J_{20}, \tau(x)\big)$ of \eqref{1slow:27} such that 
$$
\big( \phi_0(b), c_{10}(b), \tau(b) \big) = \big( \phi_0^{b,r}, c_{10}^{b,r}, b \big) ~~\textrm{and}~~\big( \phi_0(1), c_{10}(1), \tau(1) \big) = \big( \phi_0^{r}, c_{10}^{r}, 1 \big), 
$$
where $\phi_0^{b,r},~\phi_0^{r},~c_{10}^{b,r}$, and $c_{10}^{r}$ are given in Proposition \eqref{propHSb1}. It is given by
\begin{equation}\label{L5:64}
\begin{array}{l}
\phi_0(x) = \phi_0^{b,r}- \dfrac{z_1(z_1-z_2)(\phi_0^{b,r} - \phi_0^{r})}{\ln c_{10}^{b,r}- \ln c_{10}^{r}}\big(\ln c_{10}^{b,r}- \ln c_{10}(x)\big) , \\
c_{10}(x) = c_{10}^{b,r} - \dfrac{(c_{10}^{b,r} - c_{10}^{r})}{(H(1)- H(b))}(H(x)- H(b)),\\
J_{10} = \dfrac{c_{10}^{b,r}- c_{10}^{r}}{\big(H(1)-H(b)\big)} \Big(1+ \dfrac{z_1\big( \phi_0^{b,r}- \phi_0^{r}\big)}{\ln c_{10}^{b,r} - \ln c_{10}^{r}} \Big),\\
J_{20} = -\dfrac{z_1\big(c_{10}^{b,r}- c_{10}^{r}\big)}{z_2\big(H(1)-H(b)\big)} \Big(1+ \dfrac{z_2\big( \phi_0^{b,r}- \phi_0^{r}\big)}{\ln c_{10}^{b,r} - \ln c_{10}^{r}} \Big), \quad
\tau(x)= x.
\end{array}
\end{equation}
\end{lem}


\begin{proof}
It is almost identical to the proof of Lemma \ref{lem22}.
\end{proof}

To find the first order terms and for convenience we  first define the following three functions similar to those on $[b,1]$,
\begin{equation}\label{MB3-MNP}
\begin{aligned}
M_{2} :=& c_{11}^{b,r}-c_{11}^{r}+  \dfrac{\lambda z_1 - z_2}{z_2}\big( (c_{10}^{r})^2 - (c_{10}^{b,r})^2\big)\\
=& c_{10}^{b,r}w(c_{10}^{b},c_{20}^{b}) - c_{10}^{r}w(r_1,r_2) + \dfrac{\lambda z_1 - z_2}{z_2}\big( (c_{10}^{b,r})^2 - (c_{10}^{r})^2\big) + c_{10}^{b,r}\dfrac{z_1\frac{c_{21}^{b}}{c_{20}^{b}} - z_2\frac{c_{11}^{b}}{c_{10}^{b}}}{z_1-z_2},\\
N_{2}:=&  \dfrac{c_{10}^{b,r} - c_{10}^{r}}{\ln c_{10}^{b,r}- \ln c_{10}^{r}} \Bigg\{(\phi_1^{b,r} - \phi_1^{r}) 
 -\dfrac{1-\lambda}{z_2}(c_{10}^{b,r} - c_{10}^{r}) +  \dfrac{\phi_{0}^{b,r} - \phi_{0}^{r}}{c_{10}^{b,r} - c_{10}^{r}}M_{2} \\
&\qquad \qquad - \dfrac{ w(c_{10}^{b},c_{20}^{b}) - w(r_1,r_2) }{\ln c_{10}^{b,r}- \ln c_{10}^{r}}(\phi_{0}^{b,r} - \phi_{0}^{r})  -  \dfrac{\phi_{0}^{b,r} - \phi_{0}^{r}}{\ln c_{10}^{b,r}- \ln c_{10}^{r}} \dfrac{z_1 \dfrac{c_{21}^{b}}{c_{20}^{b}} -z_2 \dfrac{c_{11}^{b}}{c_{10}^{b}}}{z_1-z_2 }\Bigg\},\\
P_2(x) :=& \Big(\dfrac{z_1 \frac{c_{21}^{b}}{c_{20}^{b}} -z_2 \frac{c_{11}^{b}}{c_{10}^{b}}}{(z_1-z_2) } + w(c_{10}^{b}, c_{20}^{b}) + \frac{(\lambda z_1 -z_2)}{z_2}c_{10}^{b,r} \Big)\dfrac{\big(c_{10}^{b,r} - c_{10}(x)\big)}{c_{10}(x)}\\
 &\hspace*{.15in}+ \dfrac{\lambda z_1 -z_2}{z_2} \dfrac{\big( H(x)- H(b)\big)}{\big( H(1)- H(b)\big)}(c_{10}^{b,r} - c_{10}^{r})\\
&\hspace*{.15in}- \dfrac{M_{2}\big(H(x)- H(b)\big)}{\big(H(1)-H(b)\big)c_{10}(x)} +  \dfrac{M_{2}\big(\ln c_{10}^{b,r} - \ln c_{10}(x)\big)}{c_{10}^{b,r} - c_{10}^{r}}.
\end{aligned}
\end{equation}

\begin{lem}{\label{lem:3slow1st}} 
For $b \leq x \leq 1$, there is a unique solution $\big(\phi_1(x), c_{11}(x), J_{11}, J_{21}, \tau(x) \big)$ of \eqref{1slow:28} such that
\[
\big(\phi_1(b), c_{11}(b), \tau(b) \big) = (\phi_1^{b,r}, c_{11}^{b,r}, b)~~ \textrm{~and~}~~\big(\phi_1(1), c_{11}(1), \tau(1) \big) = (\phi	_1^{r}, c_{11}^{r}, 1),
\]
where $\phi_1^{b,r}, \phi_1^{r}, c_{11}^{b,r}$ and $c_{11}^{r}$ are given in Proposition \eqref{propHSb1}. It is given by
\begin{equation}\label{MB3-J}
\begin{array}{l}
\phi_1(x) = \phi_1^{b,r}  -\dfrac{(1-\lambda)(c_{10}^{b,r} - c_{10}^{r})}{z_2}\dfrac{ \big(H(x)-H(b)\big)}{ \big(H(1)-H(b)\big)}+ \dfrac{(\phi_{0}^{b,r} - \phi_{0}^{r})}{(\ln c_{10}^{b,r}- \ln c_{10}^{r})} P_{2}(x)\\
\hspace*{.6in} -\dfrac{\ln c_{10}^{b,r}- \ln c_{10}(x)}{c_{10}^{b,r} - c_{10}^{r}}N_{2},\\
c_{11}(x)= c_{11}^{b,r}+ \dfrac{\lambda z_1 - z_2}{z_2}\Big(c^2_{10}(x) - (c_{10}^{b,r})^2 \Big) - \dfrac{H(x)-H(b)}{H(1)-H(b)}M_{2},\\
J_{11}= \dfrac{M_{2}+z_1N_{2}}{\big(H(1)-H(b)\big)}, \hspace*{.3in} J_{21}= -\dfrac{z_1}{z_2}\dfrac{M_{2}+z_2N_{2}}{\big(H(1)-H(b)\big)}, 
\end{array}
\end{equation}
where $M_{2}$, $N_{2}$, and $P_2(x)$ were defined in \eqref{MB3-MNP}.
\end{lem}

\begin{proof} 
It is similar to the proof of Lemma \ref{lem:1slow1st}.
\end{proof}


\noindent \underline{\bf Matching for singular orbits on $[0 , 1]$ .}
The final step for forming a connecting orbit over the whole interval $[0,1]$ is to balance the three singular orbits from the preceding section at the points $x = x_1$ and $x = x_2$.
The matching conditions are $u_0^{a,l} = u_0^{a,m},~u_1^{a,l} = u_1^{a,m},~ u_0^{b,m} = u_0^{b,r},~u_1^{b,m} = u_1^{b,r}$, and $J_1= J_{10}+J_{11}d$ and $J_2= J_{20}+J_{21}d$ have to be the same on all subintervals; that is, from formulas \eqref{prop1:22}, \eqref{L2:31}, \eqref{MB1-J}, \eqref{prop2:47},  \eqref{prop2:49}, Lemma \eqref{LemM0th}, \eqref{70J_1}, \eqref{prop3:61}, \eqref{L5:64} and \eqref{MB3-J}.
We skip writing the zeroth order terms in the ion size $d$ because one can find them in previous works \cite{EL07,JLZ15}. 
First order terms, in $d$, of the governing system   are as follows,
\begin{equation*}
\begin{aligned}
&\phi_{1}^{a,m} = \phi_{1}^{a} + \dfrac{z_1 c_{10}^{a,m}}{\sigma^{a,m}}\Bigg( \dfrac{c_{11}^a}{c_{10}^a} - \dfrac{c_{21}^a}{c_{20}^a} +(1-\lambda )\Big(c_{10}^{a}+c_{20}^{a} -c_{10}^{a,m} -c_{20}^{a,m} \Big)\Bigg)\\
&\hspace*{.5in}- \dfrac{1}{\sigma^{a,m}}\Big(\dfrac{c_{21}^a}{c_{20}^a}+ 2\lambda c_{20}^{a}+ (\lambda +1) c_{10}^{a} - 2\lambda c_{20}^{a,m} - (\lambda+ 1) c_{10}^{a,m} \Big)Q,\\
&\phi_{1}^{b,m} = \phi_{1}^{b} + \dfrac{z_1 c_{10}^{b,m}}{\sigma^{b,m}}\Bigg( \dfrac{c_{11}^b}{c_{10}^b} - \dfrac{c_{21}^b}{c_{20}^b} +(1-\lambda )\Big(c_{10}^{b}+c_{20}^{b} -c_{10}^{b,m} -c_{20}^{b,m} \Big)\Bigg)\\
&\hspace*{.5in}- \dfrac{1}{\sigma^{b,m}}\Big(\dfrac{c_{21}^b}{c_{20}^b}+ 2\lambda c_{20}^{b}+ (\lambda +1) c_{10}^{b} - 2\lambda c_{20}^{b,m} - (\lambda+ 1) c_{10}^{b,m} \Big)Q,\\
&( c_{10}^{a,l}+ c_{20}^{a,l})(c_{10}^{a,l}+ \lambda  c_{20}^{a,l}) + \big( c_{11}^{a,l}+ c_{21}^{a,l}\big)
= ( c_{10}^{a,m}+ c_{20}^{a,m})(c_{10}^{a,m}+ \lambda  c_{20}^{a,m})\\
&\hspace*{2.84in} + \big(  c_{11}^{a,m}+ c_{21}^{a,m}\big)+\big( \phi_{1}^{a} - \phi_{1}^{a,m}\big) Q,\\
&( c_{10}^{b,m}+ c_{20}^{b,m})(c_{10}^{b,m}+ \lambda  c_{20}^{b,m}) + \big( c_{11}^{b,m}+ c_{21}^{b,m}\big)
= ( c_{10}^{b,r}+ c_{20}^{b,r})(c_{10}^{b,r}+ \lambda  c_{20}^{b,r})\\
&\hspace*{2.84in}+ \big(  c_{11}^{b,r}+ c_{21}^{b,r}\big)-  \big( \phi_{1}^{b} - \phi_{1}^{b,m}\big) Q,\\
&J_{11} = \dfrac{M_{0}(Q)+z_1N_{0}(Q)}{H(a)}= \dfrac{M_{2}(Q)+z_1N_{2}(Q)}{H(1)-H(b)}, \quad J_{21} =  -\dfrac{z_1}{z_2}\dfrac{M_{0}(Q)+z_2N_{0}(Q)}{H(a)}= -\dfrac{z_1}{z_2}\dfrac{M_{2}(Q)+z_2N_{2}(Q)}{H(1)-H(b)},\\
& c_{11}^{b,m}= c_{11}^{a,m}+z_1z_2{T}_{1}  S_{11}+2z_1(\lambda z_1 - z_2) {T}_{0}(Q)S_{12}(Q) + \Big( 2\lambda z_1{T}_{0}(Q)+ 2(z_1-z_2)J_{10}\Big)QS_{11} \\
&\hspace*{.5in} +z_2J_{11}y^*Q + \dfrac{z_1\Big(c_{10}^{a,m}- c_{10}^{b,m} \Big){I}_{0}(Q)c_{11}^{a,m}Q}{\sigma^{b,m}\Big(z_1{T}_{0}(Q)c_{10}^{a,m}+J_{10}Q \Big)}+ \dfrac{2z_1(\lambda z_1 - z_2){I}_{0}(Q) Q}{\sigma(y^*)}S_{12}(Q)\\
&\hspace*{.5in}+\dfrac{z_1z_2{T}_{1} {I}_{0}(Q) Q}{ {T}_{0}(Q)\sigma(y^*)}\Big( S_{11}-e^{z_1z_2{T}_{0}(Q)y^*} S_{13}(Q)\Big) - \dfrac{2z_1(z_1-z_2){I}_{0}(Q) e^{z_1z_2{T}_{0}(Q)y^*}Q}{\sigma(y^*)}  S_{14}(Q)+o(Q),
\end{aligned}
\end{equation*}
\begin{equation*}
\begin{aligned}
\phi_{1}^{b,m} =& \phi_{1}^{a,m}+\dfrac{(\lambda-1)\Big((z_1-z_2)J_{10}+{I}_{0}(Q)\Big)(c_{10}^{b,m} - c_{10}^{a,m}) }{z_2^2{T}_{0}(Q)}  - {I}_{1}y^*\\
&+ \dfrac{z_1(z_1 -z_2)\Big(c_{10}^{b,m}- c_{10}^{a,m} \Big){I}_{0}(Q)c_{11}^{a,m}}{z_2\sigma^{b,m}\Big(z_1{T}_{0}(Q)c_{10}^{a,m}+J_{10}Q \Big)} - \dfrac{(\lambda z_1 - z_2)(z_1 -z_2){I}_{0}(Q)}{z_2^2{T}_{0}(Q)\sigma^{b,m}}\big( \big(c_{10}^{b,m}\big)^2 - \big(c_{10}^{a,m}\big)^2\big)\\
&-\dfrac{z_1(z_1 -z_2){T}_{1}{I}_{0}(Q) }{ {T}_{0}(Q)\sigma^{b,m}}\Big( S_{11}-e^{z_1z_2{T}_{0}(Q)y^*} S_{13}(Q)\Big) \\
&+\dfrac{2z_1(z_1-z_2)^2{I}_{0}(Q) e^{z_1z_2{T}_{0}(Q)y^*}}{z_2\sigma^{b,m}} S_{14}(Q)+\dfrac{2(z_1-z_2)\Lambda_{0}(Q)}{z_2{T}_{0}(Q)}J_{10} y^*Q\\
 & - \dfrac{2z_1(z_1 -z_2)\Lambda_{0}(Q){I}_{0}(Q) Q }{z_2 {T}_{0}(Q)\sigma^{b,m}} S_{11}-\dfrac{(z_1-z_2)J_{11}{I}_{0}(Q) y^*Q}{ {T}_{0}(Q)\sigma^{b,m}}\\
 &+ \dfrac{(z_1-z_2)J_{11} {I}_{0}(Q) \Big(c_{10}^{b,m} -c_{10}^{a,m}\Big)}{z_2 {T}_{0}(Q)\sigma^{b,m}\big(z_1{T}_{0}(Q)c_{10}^{a,m} +J_{10}Q\big)} Q+ \dfrac{2(z_1-z_2)^2J_{10}{I}_{0}(Q) e^{z_1z_2{T}_{0}(Q)y^*}Q}{z_2{T}_{0}(Q) \sigma^{b,m}} S_{13}(Q),\\\\
\dfrac{ {T}_{1} }{(z_1-z_2){T}_{0}(Q)}&\Big((z_1-z_2)(c_{10}^{b,m} - c_{10}^{a,m}) - z_2 y^* {I}_{0}(Q) Q \Big) =  c_{11}^{b,m}- c_{11}^{a,m} - 2(\lambda z_1 - z_2) z_1 {T}_{0}(Q) S_{12}(Q)(y^*)  \\
&\hspace*{1.5in} -\Big(\big(2\lambda +\dfrac{(1-\lambda) z_2}{z_1-z_2} \big) z_1{T}_{0}(Q) + 2(z_1-z_2) J_{10}\Big)S_{11}(y^*)Q\\
& \hspace*{1.5in}  + \dfrac{z_2(\phi_{1}^{b,m} - \phi_{1}^{a,m})Q}{z_1-z_2} + 2(1-\lambda)z_1S_{15}(Q)(y^*)Q  + O(Q^2),
\end{aligned}
\end{equation*}

Now, we apply the zero-current condition on those we need in the matching system (See \eqref{Matching-zeroI}). 
We recall that $(\phi_0^{a}, c_{10}^{a}, c_{20}^{a})$, $(\phi_1^{a}, c_{11}^{a}, c_{21}^{a})$, $(\phi_0^{b}, c_{10}^{b}, c_{20}^{b})$  and $(\phi_1^{b}, c_{11}^{b}, c_{21}^{b})$  are the unknown values preassigned at $x=a$ and $x=b$, and $J_{10},J_{20},J_{11}$ and $J_{21}$ are the unknown values for the flux densities of the two types of ions.
There are also five auxiliary unknowns $\phi_0^{a,m}, \phi_1^{a,m}, \phi_0^{b,m}, \phi_1^{b,m}$, and $y^*$ in the matching system. The total number of unknowns in the matching system is twenty one, which matches the total number of equations.
Similar to \cite{EL07}, one can see the set of nonlinear equations in the matching system has a unique solution.


\section{Effects of ion size on the zero-current fluxes.}\label{Sec.Ionsize-effects}
\setcounter{equation}{0}
In this part, we are interested in reversal potential, reversal
permanent charges and zero-current fluxes determined by zero total currents. In \cite{ML19, MEL20}, the authors studied reversal potential problem for cPNP. The authors in \cite{M21} investigate the reversal permanent charge problem for cPNP.
Now, in this section, we study the connection of the zero-current quantities with the membrane potentials and diffusion constants for the hard-sphere PNP. The total current 
$I = I(V,Q)$ depends on the transmembrane potential $V$ and the permanent charge $Q$. 

The chief objective of this section is to examine the qualitative effect of ion sizes via a steady-state boundary value problem.
Note that the current $I$ depends on the size $d$, only through fluxes, i.e., $I=z_1J_1+z_2J_2$. Thus,
$$
\begin{aligned}
I = I(J_1(d), J_2(d)) =& I(J_{10}+J_{11}d+ O(d^2), J_{20}+J_{21}d+O(d^2))\\
=& I(J_{10},J_{20})+\partial_{J_1}I(J_{10},J_{20})J_{11}d + \partial_{J_2}I(J_{10},J_{20})J_{21}d + O(d^2) \\
=& I(J_{10},J_{20}) + (z_1J_{11}+z_2J_{21})d+ O(d^2) \\
:=& I_0 + I_1d + O(d^2).
\end{aligned}
$$
Therefore, $I=0$ will result in $I_0=0$ and $I_1=0$.
It follows from $I=I_0=I_1=0$ and $z_1=-z_2$  that $J_{10} =J_{20}$, $J_{11} =J_{21}$ and hence, 
$$
T_0=2J_{10},\quad T_1=2J_{11}, \quad \Lambda_0= (1+\lambda)J_{10}.
$$

One can derive the zeroth-order terms in $d$, for the case where $z_1=-z_2$, from \cite{EL07} with adding the zero-current condition $I=0$, or from \cite{ML19}  with the condition $D_1=D_2$.
To find the first order terms in $d$, we have,
\begin{equation}\label{Matching-zeroI}
\begin{aligned}
&\phi_{1}^{a,m} = \phi_{1}^{a} + \dfrac{z_1 c_{10}^{a,m}}{\sigma^{a,m}}\Bigg( \dfrac{c_{11}^{a}}{c_{10}^{a}} - \dfrac{c_{21}^{a}}{c_{20}^{a}} +(1-\lambda )\Big(c_{10}^{a}+c_{20}^{a} -c_{10}^{a,m} -c_{20}^{a,m} \Big)\Bigg)\\
&\hspace*{.5in}- \dfrac{1}{\sigma^{a,m}}\Big(\dfrac{c_{21}^{a}}{c_{20}^{a}}+ 2\lambda c_{20}^{a}+ (\lambda +1) c_{10}^{a} - 2\lambda c_{20}^{a,m} - (\lambda+ 1) c_{10}^{a,m} \Big)Q,\\
&\phi_{1}^{b,m} = \phi_{1}^{b} + \dfrac{z_1 c_{10}^{b,m}}{\sigma^{b,m}}\Bigg( \dfrac{c_{11}^{b}}{c_{10}^{b}} - \dfrac{c_{21}^{b}}{c_{20}^{b}} +(1-\lambda )\Big(c_{10}^{b}+c_{20}^{b} -c_{10}^{b,m} -c_{20}^{b,m} \Big)\Bigg)\\
&\hspace*{.5in}- \dfrac{1}{\sigma^{b,m}}\Big(\dfrac{c_{21}^{b}}{c_{20}^{b}}+ 2\lambda c_{20}^{b}+ (\lambda +1) c_{10}^{b} - 2\lambda c_{20}^{b,m} - (\lambda+ 1) c_{10}^{b,m} \Big)Q,\\
&( c_{10}^{a,l}+ c_{20}^{a,l})(c_{10}^{a,l}+ \lambda  c_{20}^{a,l}) + \big( c_{11}^{a,l}+ c_{21}^{a,l}\big)
= \big( c_{10}^{a,m}+ c_{20}^{a,m}\big)\big(c_{10}^{a,m}+ \lambda  c_{20}^{a,m}\big)\\
& \hspace*{3in}  + \big(  c_{11}^{a,m}+ c_{21}^{a,m}\big)+\big( \phi_{1}^{a} - \phi_{1}^{a,m}\big) Q,\\
&\big( c_{10}^{b,m}+ c_{20}^{b,m}\big)\big(c_{10}^{b,m}+ \lambda  c_{20}^{b,m}\big) + \big( c_{11}^{b,m}+ c_{21}^{b,m}\big)
= \big( c_{10}^{b,r}+ c_{20}^{b,r}\big)\big(c_{10}^{b,r}+ \lambda  c_{20}^{b,r}\big)\\
& \hspace*{3in} + \big(  c_{11}^{b,r}+ c_{21}^{b,r}\big)-  \big( \phi_{1}^{b} - \phi_{1}^{b,m}\big) Q,\\
&J_{11} = \dfrac{M_{0}+z_1N_{0}}{H(a)}= \dfrac{M_{2}+z_1N_{2}}{\big(H(1)-H(b)\big)}, \quad J_{21} = \dfrac{M_{0}-z_1N_{0}}{H(a)}= \dfrac{M_{2}-z_1N_{2}}{\big(H(1)-H(b)\big)},\\
& c_{11}^{b,m}=  c_{11}^{a,m}-2z_1^2{J}_{11}  S_1+4(\lambda +1)z_1^2 {J}_{10}S_2  -z_1J_{11}y^*Q +4z_1\big( \lambda + 1\big){J}_{10}S_1Q\\
&\phi_{1}^{b,m} = \phi_{1}^{a,m}+\dfrac{(\lambda-1) }{z_1}(c_{10}^{b,m} - c_{10}^{a,m}) -2(1+\lambda)J_{10} y^*Q,
\end{aligned}
\end{equation}
It follows from $J_{11}=J_{21}$ that
$M_{0}+z_1N_{0}= M_{0}-z_1N_{0}$. Hence $N_0=0$ and $J_{11}=J_{21}=M_0$, where 
$$
M_0=M_{00}+M_{01}Q.
$$
Hence $N_0=0$ and we conclude that 
\begin{equation}\label{zeroJ}
J_{11}=J_{21}= \dfrac{1}{H(a)} M_0= \dfrac{1}{H(a)}\left(M_{00}+M_{01}Q\right),
\end{equation}
with
$$
\begin{aligned}
M_{00} =&  \frac{(\lambda z_1 -z_2)}{z_1^2z_2} \alpha (r^2-l^2),\\
M_{01} =&  \frac{2(\lambda z_1 -z_2)}{z_2}c_{10,1}^{a} c_{10,0}^a+ \dfrac{\lambda}{2z_2}c_{10,0}^a +\dfrac{\lambda z_1-z_2}{2z_2(z_1-z_2)} c_{10,0}^a+ \dfrac{z_2}{z_1-z_2}(c_{11,1}^{a} +c_{21,1}^{a}),
\end{aligned}
$$
where, from the matching system, with a careful calculation,
\begin{equation}\label{c111}
\begin{aligned}
c_{11,1}^a =&  \frac{2(\lambda z_1-z_2)\alpha}{z_2(z_1-z_2)}\Phi(V)\big[(2\beta-\alpha-1)l+(\alpha^2-\beta^2)(l-r)-\beta r\big]-\frac{z_2\alpha}{z_1-z_2}(\phi_{1,0}^a-\phi_{1,0}^b) \\
& - \frac{(1-\lambda)(l-r)\alpha(\alpha-\beta)}{2z_1(z_1-z_2)}+\frac{2(\lambda z_1-z_2)}{z_2}\big((1-\alpha)c_{10,1}^ac_{10,0}^a+\alpha c_{10,1}^bc_{10,0}^b\big)\\
& +\frac{\lambda(z_1-z_2)+\lambda z_1-z_2}{2z_2(z_1-z_2)}\big((1-\alpha)c_{10,0}^a+\alpha c_{10,0}^b\big),\\
\phi_{1,0}^{a} =&  \dfrac{(\lambda z_1- z_2)(l-r)V}{z_1z_2\ln (l/r)} \Bigg( \dfrac{\alpha\big((l+r)-\alpha (l-r) \big)}{(1-\alpha)l+\alpha r }-\dfrac{2}{\ln (l/r)}\ln \dfrac{l}{(1-\alpha)l+\alpha r }\Bigg)\\
&+ \dfrac{(1-\lambda)(l-r)}{z_1z_2 \ln (l/r)}\ln \dfrac{l}{(1-\alpha)l+\alpha r }- \dfrac{(1-\lambda) \alpha (l-r) }{z_1z_2},\\
\phi_{1,0}^{b}=& \dfrac{(\lambda z_1-z_2)(l-r)V}{z_1z_2\ln (l/r)}  \Bigg( \dfrac{ (\beta -1) \big((l+r)+(1-\beta)(l-r) \big)}{(1-\beta)l+\beta r }-\dfrac{2}{\ln (l/r)} \ln\dfrac{r}{(1-\beta)l+\beta r}\Bigg)\\
&+ \dfrac{(1-\lambda)(l-r)}{z_1z_2\ln (l/r)}\ln\dfrac{r}{(1-\beta)l+\beta r}+\dfrac{(1-\lambda)(1-\beta)(l-r)}{z_1z_2},\\
c_{10,1}^a =& \dfrac{-z_2\alpha \Phi(V)}{z_1-z_2}- \dfrac{1}{2(z_1-z_2)}, \qquad c_{10,1}^b = \dfrac{z_2(1-\beta)\Phi(V)}{z_1-z_2}- \dfrac{1}{2(z_1-z_2)},
\end{aligned}
\end{equation}
with,
$$
\Phi(V) = \phi_{0,0}^a-\phi_{0,0}^b= \dfrac{V}{\ln l - \ln r} \ln \dfrac{(1-\alpha)l+\alpha r }{(1-\beta)l + \beta r},
$$
and $c_{10,0}^a=\frac{1}{z_1}((1-\alpha)l+\alpha r)$. See \cite{JLZ15} for the zeroth order terms in $Q$ and $d$. Moreover, one can find that $z_1c_{11,1}^a+z_2c_{21,1}^a=0$.
\medskip


To explore the effects of small ion size and small permanent charge on the zero-current fluxes, it can be seen from \eqref{Matching-zeroI} and from Corollary $3.6$ in \cite{LLYZ13} with $Q=0$ that
\[
 J_{11,0} + J_{21,0} = \dfrac{(\lambda z_1 - z_2)(z_2-z_1)(r^2-l^2)}{z_1^2z_2^2H(1)}= \dfrac{z_2-z_1}{\alpha z_2 H(1)}M_{00} = \dfrac{z_2-z_1}{ z_2(1-\beta) H(1)}M_{20}.
\]
and,
\[
\begin{aligned}
z_1J_{11,0} + z_2J_{21,0} =& \dfrac{(\lambda z_1 -z_2)(r-l)(z_1-z_2)}{z_1z_2H(1) (\ln r - \ln l)}\Big(\dfrac{2(r-l)}{(\ln r - \ln l)} - (r+l) \Big)V\\
&+ \dfrac{(1-\lambda)(r-l)^2(z_1-z_2)}{z_1z_2H(1) (\ln r - \ln l)}.
\end{aligned}
\]
Hence we obtain,
\begin{equation}\label{MM}
\begin{aligned}
&M_{00} = \dfrac{(\lambda z_1 - z_2)\alpha}{z_1^2z_2}(r^2-l^2), \ \ \  \ M_{20} = \dfrac{(\lambda z_1 - z_2)(1-\beta )}{z_1^2z_2}(r^2-l^2).
\end{aligned}
\end{equation}
Now we investigate the dependence of sign of $J_k$ on the membrane potential $V$ and channel geometry. 
We  analyze how  $J_{k1}$ depends on the channel geometry $(\alpha, \beta)$, and the boundary condition $(V,l,r)$.
Recall that,
\[
\begin{aligned}
J_1=& J_{10}+ J_{11} d \\
   =& J_{10,0}+J_{10,1} Q + J_{11,0}d + J_{11,1} Qd  
\end{aligned}
\]
The following Theorem is the direct conclusion of \eqref{zeroJ}---\eqref{MM}.
\begin{thm} For the zero-current flux $J_{11}$ and without any permanent charge, i.e., when $Q=0$, one has $J_{11}=J_{11,0}=M_{00}$. Therefore, if $l < r $, then $J_{11}<0$, and if $l > r $, then $J_{11}>0$. Furthermore, for non-zero small permanent charges $Q$, one has $\partial J_{11}/\partial V <0$ and,  
\begin{itemize}
\item[a.] there exist  $V^c<0$ so that for any $V<V^c$ fixed, $\partial J_{11}/\partial Q >0$, and for any $V>V^c$ fixed, $\partial J_{11}/\partial Q <0$; moreover, if $l>r$, then $J_{11}(V^c)>0$, and if $l<r$, then $J_{11}(V^c)<0$;
\item[b.] if  $l>r$, then there exists $V^c<0$ so that for any $V<V^c$, one has $J_{11}>0$, that is, the ion size $d$ increases the flux $J_1$;
\item[c.] if  $l<r$, then there exists $V^c<0$ so that for any $V>V^c$, one has $J_{11}<0$, that is,  the ion size $d$ decreases the flux $J_1$.
\end{itemize}
\end{thm}


\section{Conclusion and discussion.}\label{sec-Conclusion}\label{Sec.Conclusion}
This work investigates a one-dimensional variant of a PNP system with a local hard-sphere potential that depends on ion concentrations for ionic flow through a membrane channel with fixed boundary ion concentrations (charges) and electric potentials.
We examined a special case of ion size effects on the zero-current fluxes and their dependence on permanent charge and potentials. In \cite{FLMZ22}, the authors study the special case of $d^2$-term effects of ion size on potential and fluxes. 
The other interesting case one can study is the effects of ion size on the reversal potential. One should note that the current $I$ depends on potential $V$ and the ion size $d$ and it can be written in form of
$$
\begin{aligned}
I =& I(V,d) = I_0(V) + I_1(V) d + O(d^2)\\
  =& I_0(V_0) + \partial_VI_0(V_0)V_1d+ I_1(V_0)d+  O(d^2). 
\end{aligned}
$$
Now, $I=0$ is equivalent to $I_0(V_0) =0 $ and
$
V_1 = - \dfrac{I_1(V_0)}{\partial_VI_0(V_0)}.
$
It is challenging to evaluate $I_1$ at $V_0$ and $\partial_VI_0(V_0)$, though.\\
 We need to emphasize that the matching system gives us a huge source of investigation for the higher order terms of $Q$. One may write, for small $d$ and for general $Q$, 
\[J_k(d;Q)=J_{k0}(Q)+J_{k1}(Q)d+O(d^2),\quad k=1,2.\]

\begin{itemize}
\item[(i)] For point-charge case where $d=0$ and $J_k(0;Q)=J_{k0}(Q)$, there are several works on the effects of permanent charges on fluxes, including detailed results for $Q_0\ll 1$ and for $Q_0\gg 1$, and a general result on the flux ratio $\lambda_k(Q)=\frac{J_{k0}(Q)}{J_{k0}(0)}$; 
\item[(ii)] Also, for $Q=0$, there are results about effects of ion size $d$ on fluxes;
\item[(iii)] One can easily obtain results about effects of leading terms in small $Q_0$ and small $d$ based on the above result. 

\end{itemize}

\noindent We are generally interested in effects of  interactions between $Q$ and small $d$; that is, $J_{k1}(Q)d$. 
One may consider small $Q$ so that
\[J_{k1}(Q)d=\big(J_{k1,0}+J_{k1,1}Q+O(Q^2)\big)d=J_{k1,0}d+J_{k1,1}Qd+O(Q^2)d.\]
Note that $J_{k1,0}d$ is precisely the leading terms of ion size effects when $Q=0$ mentioned in (ii). We thus focus on, here, the term $J_{k1,1}Qd$.  More precisely, 
 one interesting question is to know how the factors $J_{k1,1}$'s depend on the boundary conditions $(V, l, r)$ and  the value $\lambda$ as well as the key geometry $\big(H(a), H(b), H(1)\big)$ of the channel. In particular, we are interested in the signs of $J_{k1,1}$'s in terms of the above parameters.
It follows from the formulas for $J_{k1,1}$'s in \eqref{Matching-zeroI} that
\[
J_{k1,1} >0 \;\mbox{ if and only if }\; M_{01}+z_kN_{01}>0.
\]
Furthermore, recalling $I_1$ and $T_1$ defined in \eqref{Abbrev},
\[
I_1=z_1J_{11,1}+z_2J_{21,1}=\frac{z_1(z_1-z_2)}{H(a)}N_{01}\;\mbox{ and }\; T_1=J_{11,1}+J_{21,1}=-\frac{z_1-z_2}{z_2H(a)}M_{01},
\]
then $I_1Qd$ is the corresponding current (total flux of charges) term and $T_1Qd$ is the corresponding total flux of matter term. One has 
\[
I_1>0  \;\mbox{ if and only if }\; N_{01}>0;\quad T_1>0  \;\mbox{ if and only if }\; M_{01}>0.
\]
Then one may determine the  conditions for 
\[M_{01}+z_kN_{01}>0, \quad M_{01}>0,\quad N_{01}>0\]
in terms of 
\[\big(V, l, r, \lambda, H(a), H(b), H(1)\big).\]
 It may be possible to solve for, say critical values of $V$, from each of the  equations
\[M_{01}+z_kN_{01}=0, \quad M_{01}=0,\quad N_{01}=0\] 
in terms of the other variables $\big(l, r, \lambda, H(a), H(b), H(1)\big)$.
\section{Appendix: Proof of Lemma \ref{lem:2slow1st}  } \label{sec-Appendix}
\setcounter{equation}{0}
To find the first order terms, multiplying $h(\tau)\sigma (\tau)$ to the right hand of the system \eqref{Mfirst}, one has,
\begin{equation}\label{2sB:4.1}
\begin{aligned}
\dfrac{d}{dy}{\phi_1}=&  \dfrac{{I}_0 }{\sigma(y)}\Big( (z_1 - z_2) z_1 c_{11} + 2(1-\lambda)z_1c_{10}Q - 2\lambda Q^2\Big)\\
& + \Big((1-\lambda)z_1c_{10} - \lambda Q\Big){T}_0 - {I}_1 - \Lambda_0 Q,\\
\dfrac{d}{dy}{c}_{11}= & 2(\lambda z_1 - z_2) z_1c_{10}^2 {T}_0 + z_1z_2c_{10}{T}_1 \\
& +   \Big(\big[2\lambda z_1 {T}_0 +2(z_1-z_2)J_{10}\big]c_{10} + z_2J_{11} \Big) Q\\
& - \dfrac{ {I}_0 Q }{\sigma(y)}\Big( z_1 z_2c_{11} + z_1c_{10}\big( 2(1-\lambda)z_1c_{10} - 2\lambda Q \big)\Big),\\
\dfrac{d}{dy}{J}_{11} = & \dfrac{d}{dy}{J}_{21} = 0, \quad \dfrac{d}{dy}{\tau} =h(\tau)\sigma(\tau).
\end{aligned}
\end{equation}
We rewrite the second equation in \eqref{2sB:4.1} in linear form,
\begin{equation}\label{2sB:61}
\begin{aligned}
\dfrac{d}{dy}{c}_{11}= -Z_0(y)c_{11}+Z_1(y),
\end{aligned}
\end{equation}
where  we defined $Z_0(y)$ in \eqref{Abbrev} and,
\begin{equation}\label{2sB:62}
\begin{aligned}
Z_1(y)=& 2(\lambda z_1 - z_2) z_1c_{10}^2 {T}_0 + z_1z_2c_{10}{T}_1 +  \Bigg(\Big(2\lambda z_1 {T}_0 +2(z_1-z_2) J_{10}\Big)c_{10} + z_2J_{11} \Bigg) Q\\
& - \dfrac{2z_1c_{10}{I}_0 Q }{\sigma(y)}\Big((1-\lambda)z_1c_{10} - \lambda Q \Big).
\end{aligned}
\end{equation}
The solution of the linear differential equation in \eqref{2sB:61} is,
\begin{equation}\label{c11y}
\begin{aligned}
 c_{11}(y) = c_{11}^{a,m}e^{-\int_0^y Z_0(s)ds} + e^{-\int_0^y Z_0(s)ds} \int_0^y Z_1(s)e^{\int_0^s Z_0(v)dv}ds.
\end{aligned}
\end{equation}
In order to obtain $c_{11}(y)$ in above, we first find $e^{\int_0^y Z_0(s)ds}$. It follows from \eqref{Int-Z}, \eqref{Abbrev}, and from Lemma \eqref{LemM0th} that,
$$
\begin{aligned}
e^{\int_0^y Z_0(s)ds}&= \dfrac{\sigma(y)}{\sigma^{a,m}}e^{-z_1z_2{T}_0 y}=\dfrac{1}{{T}_0 \sigma^{a,m}}\Big((z_1 -z_2)\big(z_1c_{10}^{a,m} {T}_0 + J_{10}Q \big) - {I}_0 Qe^{-z_1z_2{T}_0 y} \Big).
\end{aligned}
$$
Then, from above and from Lemma \ref{LemM0th} one has,
\begin{equation}\label{eJ-pn}
\begin{aligned}
&e^{-\int_0^y Z_0(s)ds}= \dfrac{\sigma^{a,m}}{\sigma(y)}e^{z_1z_2{T}_0 y},\quad e^{z_1z_2{T}_0 y}= \dfrac{z_1{T}_0c_{10}(y) +J_{10}Q}{z_1{T}_0 c_{10}^{a,m} +J_{10}Q}.
\end{aligned}
\end{equation}
Thus, from \eqref{eJ-pn} and  \eqref{2sB:62},
$$
\begin{aligned}
\int_0^y Z_1(s)e^{\int_0^s Z_0(v)dv}ds = & \dfrac{2z_1(\lambda z_1 - z_2) (z_1 -z_2) }{\sigma^{a,m}}\Big(z_1c_{10}^{a,m} {T}_0 + J_{10}Q \Big) S_2 - \dfrac{2z_1(\lambda z_1 - z_2) {I}_0 Q}{\sigma^{a,m}}S_4\\
& +\dfrac{z_1z_2{T}_1(z_1 -z_2)}{{T}_0\sigma^{a,m}}\Big(z_1c_{10}^{a,m}{T}_0 + J_{10}Q \Big) S_1  -\dfrac{z_1z_2{T}_1 {I}_0 Q}{{T}_0 \sigma^{a,m}} S_3\\
&+ \dfrac{\big( 2\lambda z_1{T}_0  + 2(z_1-z_2)J_{10}\big) (z_1 -z_2)Q}{{T}_0 \sigma^{a,m}}\Big(z_1c_{10}^{a,m} {T}_0 + J_{10}Q \Big) S_1\\
&  - \dfrac{{I}_0Q^2}{{T}_0 \sigma^{a,m}}\big( 2\lambda z_1{T}_0  + 2(z_1-z_2)J_{10}\big) S_3 + \dfrac{z_2J_{11}  (z_1 -z_2)Q}{{T}_0\sigma^{a,m}} \Big(z_1c_{10}^{a,m} {T}_0 + J_{10}Q \Big) y   \\
&  + \dfrac{J_{11} {I}_0 Q^2}{z_1{T}_0^2\sigma^{a,m}}  \big(e^{-z_1z_2{T}_0y}-1 \big) - \dfrac{2{I}_0Q}{\sigma^{a,m}}\Big((1-\lambda)z_1^2 S_4 - \lambda Qz_1 S_3  \Big),
\end{aligned}
$$
where $S_z$ and $S_{k}(y)$'s, for $k=1,\cdots,5$,  were defined in \eqref{Sz} and \eqref{S-integ} respectively.
Hence, from \eqref{c11y} and above computations we obtain,
\begin{equation}\label{MB2-c11}
\begin{aligned}
c_{11}(y) =& c_{11}^{a,m}\dfrac{\sigma^{a,m}}{\sigma(y)}e^{z_1z_2{T}_{0}y}  + \dfrac{z_1z_2{T}_{1}(z_1 -z_2)}{{T}_{0}\sigma(y)} \Big(z_1 {T}_{0}c_{10}(y) + J_{10}Q \Big) S_1 \\
& +\dfrac{2z_1(\lambda z_1 - z_2) (z_1 -z_2) }{\sigma(y)}\Big(z_1 {T}_{0} c_{10}(y)+ J_{10}Q \Big) S_2  - \dfrac{2z_1(\lambda z_1 - z_2) {I}_{0} Q}{\sigma(y)}e^{z_1z_2{T}_{0}y}S_4\\
& -\dfrac{z_1z_2{T}_1 {I}_0 Q}{{T}_0 \sigma(y)}e^{z_1z_2{T}_0 y} S_3 + \dfrac{  (z_1 -z_2)Q}{{T}_0 \sigma(y)}\big( 2\lambda z_1{T}_0 + 2(z_1-z_2)J_{10}\big)\big(z_1 {T}_0 c_{10}(y) + J_{10}Q \big) S_1\\
&  - \dfrac{\big( 2\lambda z_1{T}_0  + 2(z_1-z_2)J_{10}\big ){I}_0 Q^2}{{T}_0 \sigma(y)}e^{z_1z_2{T}_0 y} S_3  + \dfrac{z_2J_{11}  (z_1 -z_2)Q}{{T}_0 \sigma (y)}\Big(z_1 {T}_0 c_{10}(y) + J_{10}Q \Big)  y   \\
& + \dfrac{J_{11} {I}_0 Q^2}{z_1{T}_0^2\sigma(y)} \big(1 -e^{z_1z_2{T}_0y} \big)  + \dfrac{2{I}_0 e^{z_1z_2{T}_0y}Q}{\sigma(y)}\Big(\lambda z_1Q S_3 - (1-\lambda)z_1^2 S_4\Big).
\end{aligned}
\end{equation}
On the other hand,
\begin{equation}\label{MB2-simp}
\begin{aligned}
\dfrac{ z_1{T}_0 c_{10}(y)+J_{10}Q}{\sigma(y)} =&\dfrac{{T}_0}{z_1-z_2}+ \dfrac{{I}_0Q}{(z_1-z_2)\sigma(y)}, \\
 \dfrac{\sigma^{a,m}}{z_1{T}_0c_{10}^{a,m}+J_{10}Q } =& \dfrac{z_1-z_2}{{T}_0}- \dfrac{{I}_0 Q}{{T}_0 \big(z_1{T}_0c_{10}^{a,m}+J_{10}Q \big)},\\
\dfrac{\sigma^{a,m}}{\sigma(y)}e^{z_1z_2{T}_{0}y} =&\dfrac{\sigma^{a,m}}{\sigma(y)}\dfrac{z_1{T}_0 c_{10}(y) +J_{10}Q}{z_1{T}_0 c_{10}^{a,m} +J_{10}Q}\\
=& 1 + \dfrac{{I}_0 Q}{{T}_0 \sigma(y)}-\dfrac{{I}_0 Q}{(z_1-z_2)\big(z_1{T}_0 c_{10}^{a,m}+J_{10}Q \big)}\\
& - \dfrac{{I}_0^2Q^2}{(z_1-z_2){T}_0\sigma(y)\big(z_1{T}_0 c_{10}^{a,m}+J_{10}Q \big)}\\
=& 1 + \dfrac{z_1\big(c_{10}^{a,m}- c_{10}(y) \big){I}_0 Q}{\sigma(y)\big(z_1{T}_0 c_{10}^{a,m}+J_{10}Q \big)}.
\end{aligned}
\end{equation}

\begin{rem} As we go further on equation \eqref{2sB:27}, $Q$ appears to the denominator of some terms of $\phi_1(y)$. 
Thus, we need to keep $Q^2$ terms of $c_{11}(y)$. Moreover, in order to simplify $\phi_1(y)$ terms, we need to manipulate the corresponding terms in $c_{11}(y)$ applying the formulas in \eqref{MB2-simp}.
\end{rem}
 Now, applying the equations in \eqref{MB2-simp} on $c_{11}(y)$ in \eqref{MB2-c11}, the desired expression for  $c_{11}(y)$ will be obtained.
To find $\phi_1(y)$, we first need to obtain $\int_0^yZ_1(s)ds$. Then, from \eqref{Abbrev}, and \eqref{2sB:61} one has
\begin{equation}\label{2sB:26}
\begin{aligned}
\int_0^y \dfrac{{I}_0 (z_1-z_2)z_1c_{11}(s)}{\sigma(s)} ds
= & \dfrac{z_1 -z_2}{z_2Q} \int_0^y Z_0(s)c_{11}(s) ds\\
=& \dfrac{z_1 -z_2}{z_2Q} \int_0^y \Big(Z_1(s)- \frac{d}{dy}{c}_{11}(s)\Big) ds\\
 =&\dfrac{z_1 -z_2}{z_2Q}\Big( \int_0^yZ_1(s)ds - \big(c_{11}(y)- c_{11}^{a,m}\big) \Big).
\end{aligned}
\end{equation}
\\
Now, from $\phi_1$-equation in \eqref{2sB:4.1}, and relations \eqref{Int-Z} , \eqref{2sB:26} and the integrals on \eqref{S-integ},
\begin{equation}\label{2sB:27}
\hspace*{-.1in}\begin{aligned}
\phi_1(y) = & \phi_1^{a,m}+ \dfrac{z_1 -z_2}{z_2Q}\Big(z_1z_2{T}_1 S_1 + 2(\lambda z_1 - z_2) z_1{T}_0 S_2 \Big)\\
&+\dfrac{(z_1 -z_2)}{z_2}\Big(\big[2\lambda z_1 {T}_0 + 2(z_1-z_2) J_{10}\big]S_1 + z_2yJ_{11} \Big)\\
& +\dfrac{(z_1 -z_2)}{z_2}\Big(2\lambda Q S_z(y)- 2(1-\lambda)z_1S_5\Big) \\
&-\dfrac{(z_1 -z_2)}{z_2Q}\Big (c_{11}(y)- c_{11}^{a,m} \Big)+ 2(1-\lambda)Q S_z(y)\\
& - \dfrac{2\lambda Q}{z_1z_2}\Big( \ln \sigma(y) - \ln \sigma^{a,m} - z_1z_2{T}_0 y\Big) \\
& + {T}_0 \Big( (1-\lambda) z_1S_1 - \lambda y Q\Big)- {I}_1 y - \Lambda_0 yQ.
\end{aligned}
\end{equation}
Now, substitute $c_{11}(y)$  into the above, use \eqref{eJ-pn}, and carefully simplify to obtain,
$$
\begin{aligned}
\phi_1(y) =& \phi_1^{a,m}+ (1-\lambda) z_1{T}_0 \dfrac{c_{10}(y) - c_{10}^{a,m} -z_2J_{10}Qy}{z_1z_2{T}_0} - {I}_1 y - \dfrac{2(1-\lambda){I}_0 \big(c_{10}(y) - c_{10}^{a,m}\big)}{z_2^2{T}_0}\\
 & -  \dfrac{2(1-\lambda)Q}{z_1(z_1 -z_2)}\Big( \ln \sigma(y) - \ln \sigma^{a,m}\Big) + \dfrac{2(1-\lambda)(z_1 -z_2)J_{10}^2yQ}{z_2{T}_0}- \dfrac{z_1(z_1 -z_2)\Big(c_{10}^{a,m}- c_{10}(y) \Big){I}_0 c_{11}^{a,m}}{z_2\sigma(y) \Big(z_1{T}_0 c_{10}^{a,m}+J_{10}Q \Big)} \\
&- \dfrac{2z_1(z_1 -z_2)(\lambda z_1 - z_2){I}_0}{z_2\sigma(y)}\Big(\dfrac{c_{10}^2(y) - \big(c_{10}^{a,m}\big)^2-2z_2J_{10}S_1Q }{2z_1z_2{T}_0 }\Big) -\dfrac{z_1(z_1 -z_2){T}_1 {I}_0}{{T}_0 \sigma(y)}\Big( S_1-e^{z_1z_2{T}_0 y} S_3\Big)\\
&  +\dfrac{2z_1(z_1-z_2)^2{I}_0 e^{z_1z_2{T}_0 y}}{z_2\sigma(y)} S_4 - \dfrac{(z_1 -z_2)\Big( 2\lambda z_1{T}_0 + 2(z_1-z_2)J_{10}\Big){I}_0 Q }{z_2{T}_0\sigma(y)} S_1 -\dfrac{(z_1-z_2)J_{11}{I}_0 yQ}{{T}_0 \sigma(y)} \\
& - \dfrac{(z_1-z_2)J_{11} {I}_0 \Big(c_{10}^{a,m} -c_{10}(y)\Big)}{z_2{T}_0 \sigma(y) \big(z_1{T}_0 c_{10}^{a,m} +J_{10}Q\big)} Q + \dfrac{2(z_1-z_2)^2J_{10}{I}_0 e^{z_1z_2{T}_0 y}Q}{z_2{T}_0 \sigma(y)} S_3\\
&+ \dfrac{2}{z_2}\Big((1-\lambda)z_2+\lambda(z_1 -z_2) \Big)Q\Big(J_{10} y +\dfrac{ \ln \sigma(y) - \ln \sigma^{a,m}}{z_1(z_1-z_2)}\Big) - \dfrac{2\lambda Q}{z_1z_2}\Big( \ln \sigma(y) - \ln \sigma^{a,m} \Big) +(\lambda -1) J_{10}yQ.
\end{aligned}
$$


The proof is complete now considering that \cite{JLZ15},
\[
\begin{aligned}
& J_{11,0} + J_{21,0} = \dfrac{(\lambda z_1 - z_2)(z_1-z_2)(l^2-r^2)}{z_1^2z_2^2H(1)},\hspace*{.3in} J_{10,1}+J_{20,1}= \dfrac{\phi_{0,0}^b-\phi_{0,0}^a}{H(1)},\\
&J_{10,0}+J_{20,0} = \dfrac{(z_2-z_1)(l-r) }{z_1z_2H(1)}.
\end{aligned} 
\]

\noindent 
{\bf Acknowledgement.}  Hamid Mofidi thanks Mingji Zhang and Jianing Chen for productive communications.


  \bibliographystyle{plain}

\begin{thebibliography}{999} 

 


{
}      



\bibitem{B42} J.J. Bikerman, 
{\em  Structure and capacity of the electrical double layer.} 
 Philos. Mag. 33 (1942), 384.
  
\bibitem{BNVHEG} D. Boda, W. Nonner, M. Valisko, D. Henderson, B. Eisenberg, and D. Gillespie,
Steric selectivity in Na channels arising from protein polarization and mobile side
chains. 
{\em Biophys. J.} {\bf 93} (2007), 1960-1980.


\bibitem{BKSA09} M. Z. Bazant, M. S. Kilic, B. D. Storey, and A. Ajdari, Towards an understanding of induced charge
electrokinetics at large applied voltages in concentrated solutions. 
 {\em Adv. Coll. Interf. Sci.} {\bf 152} (2009),   48--88.



 


\bibitem{CE} D. P. Chen  and R.S. Eisenberg, 
Charges, currents and potentials in ionic channels
of one conformation. 
{\em Biophys. J.} {\bf 64} (1993),   1405-1421.



%


\bibitem{Eis} B. Eisenberg,  
 Ion Channels as Devices.
{\em J. Comp. Electro.} {\bf 2} (2003),   245-249.
 
\bibitem{Eis00} B. Eisenberg, 
Crowded charges in ion channels. 
{\em Advances in Chemical Physics} (ed. S. A. Rice) (2011), 77-223, 
John Wiley and Sons, Inc. New York.

 

\bibitem{EHL10} B. Eisenberg, Y. Hyon, and C. Liu,
Energy variational analysis of ions in water and channels:
Field theory for primitive models of complex ionic fluids.
{\em J. Chem. Phys.} {\bf 133} (2010),  104104(1-23).

 

\bibitem{EL07} B. Eisenberg and W. Liu,
Poisson-Nernst-Planck systems for ion channels with permanent charges. 
{\em SIAM J. Math. Anal.} {\bf 38} (2007),   1932-1966.


 \bibitem{EL17} B. Eisenberg and W. Liu,
Relative dielectric constants and selectivity ratios in open ionic channels. 
{\em Mol. Based Math. Biol.} {\bf 5} (2017),   125-137.

 \bibitem{FLMZ22} Y. Fu, W. Liu, H. Mofidi and M. Zhang, Finite Ion Size Effects on Ionic Flows via Poisson-Nernst-Planck Systems: Higher Order Contributions. 
{\em J.  Dynam. Differential Equations}  (2022),   1-25.


%

%



 
 



\bibitem{GNE} D. Gillespie,  W. Nonner,  and R. S. Eisenberg,  
Coupling Poisson-Nernst-Planck and density functional theory
to calculate ion flux.
{\em J. Phys.: Condens. Matter} {\bf 14} (2002),   12129-12145.



\bibitem{HBRR18} S. M. H. Hashemi Amrei, S. C. Bukosky, S. P. Rader, W. D.
Ristenpart, and G. H. Miller, Oscillating Electric Fields in
Liquids Create a Long-Range Steady Field. 
{\em  Phys. Rev. Lett.} {\bf 121}  (2018) 185504.	

\bibitem{Hek} G. Hek. Geometric singular perturbation theory in biological practice. {\em J. Math. Biol.} {\bf 60} (2010), 347-386.


\bibitem{Hil01} B. Hille,   {\em Ion Channels of Excitable Membranes (Third Edition).}
 Sinauer Associates, Inc., Sunderland, Massachusetts, USA 2001.

\bibitem{Hille89} B. Hille,  Transport Across Cell Membranes: Carrier Mechanisms, Chapter 2. Textbook of Physiology (ed. H. D. Patton, A. F. Fuchs, B. Hille, A. M. Scher and R. D. Steiner).
           Philadelphia, Saunders {\bf 1} (1989), 24-47.
%
%
 \bibitem{HEL10} Y. Hyon, B. Eisenberg,  and C. Liu,
A mathematical model for the hard sphere repulsion in ionic solutions.
 {\em Commun. Math. Sci.} {\bf 9} (2010),  459-475.
%


 


\bibitem{IR} W. Im  and B. Roux, 
Ion permeation and selectivity of
OmpF porin: a theoretical study based on molecular dynamics,
Brownian dynamics, and continuum electrodiffusion theory.
{\em J. Mol. Biol.} {\bf 322} (2002),    851-869.
 
 \bibitem{JEL19} S. Ji, B. Eisenberg, and W. Liu,
 Flux Ratios and Channel Structures.
 {\em J. Dynam. Differential Equations} (2019),pp.1141-1183, https://doi.org/10.1007/s10884-017-9607-1.
 
\bibitem{JL12} S. Ji and W. Liu,
Poisson-Nernst-Planck systems for ion flow with density functional
theory for  hard-sphere potential: I-V relations and critical potentials. Part I: Analysis.  
{\em J. Dynam. Differential Equations} {\bf 24} (2012), 955-983. 
  
 \bibitem{JLZ15} S. Ji,  W. Liu, and M. Zhang,
Effects of (small) permanent charge and channel geometry on ionic flows
via classical Poisson-Nernst-Planck models,
{\em SIAM J. Appl. Math.} {\bf 75} (2015), 114-135.

\bibitem{Jones95}  C. Jones, Geometric singular perturbation theory, in: dynamical systems, montecatini terme, 1994, in: Lecture Notes in
Math. {\em Springer-Verlag, Berlin,} {\bf 1609}  (1995), 44-118.

\bibitem{JK94} C. Jones and N. Kopell,  
Tracking invariant manifolds with differential forms in singularly perturbed systems. 
{\em J. Differential Equations} {\bf 108} (1994),  64-88. 
\bibitem{Kuehn15} C. Kuehn. Multiple Time Scale Dynamics, Springer, 2015.
%


\bibitem{LLYZ13} G. Lin, W. Liu,  Y. Yi, and M. Zhang,
Poisson-Nernst-Planck systems  for ion flow with a local hard-sphere
potential for ion size effects.
{\em SIAM J. Appl. Dyn. Syst.} {\bf 12} (2013), 1613-1648.


 



\bibitem{Liu00} W. Liu, Exchange lemmas for singular perturbations with certain turning points. {\em J. Differential Equations} {\bf 167} (2000), 134-180.
 
\bibitem{Liu05} W. Liu, 
Geometric singular perturbation approach to steady-state Poisson-Nernst-Planck systems.
{\em SIAM J. Appl. Math.} {\bf 65} (2005),   754-766.

\bibitem{Liu09} W. Liu, 
One-dimensional steady-state Poisson-Nernst-Planck systems for ion channels with multiple ion species. 
{\em J. Differential Equations} {\bf 246} (2009),    428-451.





 

\bibitem{LW10} W. Liu  and B. Wang,  
Poisson-Nernst-Planck systems for narrow tubular-like membrane channels. 
{\em J. Dynam. Differential Equations} {\bf 22} (2010),    413-437.







\bibitem{M21} H. Mofidi, Reversal permanent charge and concentrations in ionic flows via Poisson-Nernst-Planck models.
{\em Quart. Appl. Math.} {\bf 79} (2021),  581-600.


\bibitem{ML19} H. Mofidi and W. Liu, 
Reversal potential  and reversal permanent charge with unequal diffusion coefficients via classical Poisson--Nernst--Planck models.
{\em  SIAM J. Appl. Math. } {\bf 80} (2020), 1908-1935.


 
\bibitem{MEL20} H. Mofidi, B. Eisenberg and W. Liu, 
Effects of Diffusion Coefficients and Permanent Charge on Reversal Potentials in Ionic Channels.
{\em  Entropy} {\bf 22} (2020), 325(1-23).






\bibitem{NE98} W. Nonner  and R. S.  Eisenberg,  
Ion permeation and glutamate residues linked by Poisson-Nernst-Planck theory in L-type Calcium channels.
{\em Biophys. J.} {\bf 75} (1998),   1287-1305.







	

 















\bibitem{SNE01} Z. Schuss, B. Nadler,  and R. S.  Eisenberg,
Derivation of Poisson and Nernst-Planck equations in a bath and
channel from a molecular model.
{\em Phys. Rev. E} {\bf 64} (2001),    1-14.

  
  \bibitem{SL18} L. Sun and W. Liu, Non-localness of excess potentials and boundary value problems of Poisson-Nernst-Planck systems for ionic flow: A Case Study.
  {\em J. Dynam. Differential Equations} {\bf 30} (2018), 779-797.
  

\bibitem{VGK10} D. Vasileska,  S. M. Goodnick and G. Klimeck. {\em Computational Electronics: Semiclassical and Quantum Device Modeling and Simulation. }
{ New York, CRC Press}, 2010. 
   





  
  


 \end{thebibliography}

\end{document}